\documentclass[12pt, draftclsnofoot, onecolumn]{IEEEtran}
\usepackage[pdftex]{graphicx}
\usepackage{texdraw,epsfig,bm}
\usepackage{amsfonts}
\usepackage{amsmath}
\usepackage{amssymb}
\usepackage{amsthm}
\usepackage{xcolor}
\usepackage{bbm}
\usepackage{pgfplots} 
\usepackage{tikzscale}
\usepackage[ruled,norelsize]{algorithm2e}
\usepackage{epstopdf}
\usepackage{tikz}
\ifCLASSOPTIONcompsoc
    \usepackage[caption=false, font=normalsize, labelfont=sf, textfont=sf]{subfig}
\else
\usepackage[caption=false, font=footnotesize]{subfig}
\fi

\usetikzlibrary{calc}
\allowdisplaybreaks 

\newtheorem{thm}{Theorem}

\newtheorem{lem}{Lemma}

\newtheorem{prb}{Problem}

\newcommand{\ie}{{\it i.e.},\ }

\newcommand{\BS}[1]{{\color{black} #1}}

\newcommand{\JG}[1]{{\color{black} #1}}
\newcommand{\BSg}[1]{{\color{black} #1}}

\makeatletter
\newcommand{\removelatexerror}{\let\@latex@error\@gobble}
\makeatother

\begin{document}
\title{Optimal Geographical Caching in Heterogeneous Cellular Networks}
\author{Berksan Serbetci, and Jasper Goseling
\thanks{A part of this work~\cite{wcncpaper} has been presented at the IEEE Wireless Communications and Networking Conference 2017, held at San Francisco, CA, USA.}
\thanks{B. Serbetci was with the Stochastic Operations Research, University of Twente, Enschede, 7522 NB, The Netherlands. He is now with the Communication Systems Department, EURECOM, 06410, Biot, France (e-mail: berksan.serbetci@eurecom.fr).}
\thanks{J. Goseling is with the Stochastic Operations Research, University of Twente, Enschede, 7522 NB, The Netherlands. (e-mail: j.goseling@utwente.nl).}
}
\maketitle
\begin{abstract}
We investigate optimal geographical caching in heterogeneous cellular networks, where different types of base stations (BSs) have different cache capacities. The content library contains files with different popularities. The performance metric is the total hit probability. 

The problem of optimally placing content in all BSs jointly is not convex in general. However, we show that when BSs are deployed according to homogeneous Poisson point processes (PPP), independently for each type, we can formulate the problem as a convex problem. We give the optimal solution to the joint problem for PPP deployment. For the general case, we provide a distributed local optimization algorithm (LOA) that finds the optimal placement policies for different types of BSs. We find the optimal placement policy of the small BSs (SBSs) depending on the placement policy of the macro BSs (MBSs). We show that storing the most popular content in the MBSs is almost optimal if the SBSs are \JG{using an optimal} placement policy. Also, for the SBSs no such heuristic can be used; the optimal placement is significantly better than storing the most popular content. Finally, we numerically verify that LOA gives the same hit probability as the joint optimal solution for the PPP model.
\end{abstract}
\section{Introduction}
In recent years, there is extreme growth in data traffic over cellular networks. The growth rate of the demand is expected to increase in the upcoming years~\cite{cisco} such that current network infrastructures will not be able to support this data traffic~\cite{femtocells}. In order to tackle this problem, an obvious approach is to increase the number of base stations. These base stations require a high-speed backhaul to make this system work properly and it is costly to connect every base station to the core network in real life. A solution to this problem is to reduce backhaul traffic by reserving some storage capacity at both macro base stations (MBSs) and small base stations (SBSs) and use these as caches~\cite{femtod2d}. In this way, part of the data is stored at the wireless edge and the backhaul is used only to refresh this stored data. Data replacement will depend on the users' demand distribution over time. Since this distribution is varying slowly, the stored data can be refreshed at off-peak times. In this way, caches containing popular content changing over time depending on the demand will serve as helpers to the overall system and decrease the backhaul traffic.

Recently, there has been growing interest in caching in cellular networks.
In~\cite{femtocaching} Shanmugam \emph{et al.} focus on the content placement problem and analyze which files should be cached by which helpers for the given network topology and file popularity distribution by minimizing the expected total file delay.
In~\cite{approximation} Poularakis \emph{et al.} provide an approximation algorithm for the problem of minimizing the user content requests routed to macrocell base stations with constrained cache storage and bandwidth capacities.
In~\cite{optimalgeographic} B{\l}aszczyszyn \emph{et al.} revisit the optimal content placement in cellular caches by assuming a known distribution of the coverage number and provide the optimal probabilistic placement policy which guarantees maximal total hit probability.
In~\cite{fundamental} Maddah-Ali \emph{et al.} developed an information-theoretic lower bound for the caching system for local and global caching gains. In~\cite{distrcach} Ioannidis \emph{et al.} propose a novel mechanism for determining the caching policy of each mobile user that maximize the system's social welfare.
In~\cite{exploiting} Poularakis \emph{et al.} consider the content storage problem of encoded versions of the content files. In~\cite{cacheenabled} Bastug \emph{et al.} couple the caching problem with the physical layer. In~\cite{diststorage} Altman \emph{et al.} compare the expected cost of obtaining the complete data under uncoded and coded data allocation strategies for caching problem. Cache placement with the help of stochastic geometry and optimizing the allocation of storage capacity among files in order to minimize the cache miss probability problem is presented by Avrachenkov \emph{et al.} in~\cite{optimizationofcaching}. A combined caching scheme where part of the available cache space is reserved for caching the most popular content in every small base station, while the remaining is used for cooperatively caching different partitions of the less popular content in different small base stations, as a means to increase local content diversity is proposed by Chen \emph{et al.} in~\cite{cooperativecaching}. In~\cite{utility} Dehghan \emph{et al.} associate with each content a utility, which is a function of the corresponding content hit probabilities and propose utility-driven caching, where they formulate an optimization problem to maximize the sum of utilities over all contents.

The \emph{main contribution} of this paper is to find optimal placement strategies that maximize total hit probability in heterogeneous cellular networks. \BS{Different from~\cite{optimalgeographic, optimizationofcaching} our focus is on heterogeneous cellular networks in which an operator wants to jointly optimize the cached content in macro base stations (MBSs) and small base stations (SBSs) with different storage capacities and different coverage radii.} This problem is not convex in general conditions. We show that it is possible to reformulate the problem and make it convex when base stations are deployed according to homogeneous Poisson point processes (PPP). \BS{To our knowledge, this is the first paper that provides an optimal solution to the optimization of heterogeneous (and multi-tier) caching devices. We show that optimal placement strategies of the base stations can be flexibly distributed over different types if the sum of the file placement probabilities times the density parameters of the base stations satisfy a certain capacity constraint.} As the general problem is not convex, we provide a distributed local optimization algorithm (LOA) and optimize only one type of cache (e.g. SBS) using the information coming from other types of caches (e.g. MBS and other SBSs with different cache capacities) at each iteration step. We numerically verify that for PPP deployment scenario, LOA converges to the optimal hit probability that is found by solving the joint convex optimization problem after one round. We also illustrate with numerical examples how LOA performs for non-PPP deployment scenarios.

For several configurations we show that whether MBSs use the optimal deployment strategy or store ``the most popular content", has no impact on the total hit probability after deploying the SBSs with optimal content placement policies. We show that it is crucial to optimize the content placement strategy of the SBSs in order to maximize the overall performance. We show that heuristic policies like storing the popular content that is not yet available in the MBSs result in significant performance penalties.

The placement strategies that are proposed in this paper are probabilistic in nature. Therefore, they provide a very low-complexity solution to content placement in large networks. In~\cite{gibbsian} and~\cite{sigmetrics} non-probabilistic strategies are proposed that take into account exact base station locations and the overlap in coverage regions. These strategies will result in higher hit probabilities, but come at significantly larger complexity.

The remainder of this paper is organized as follows. In Section~\ref{model} we start the paper with model and problem definition. In Section~\ref{sec:jointoptimization} we present the joint optimal placement strategy problem for the PPP model and give required tools to solve it. In Section~\ref{sec:distributedoptimization} we provide a distributed local optimization algorithm for the general non-convex joint optimization problem. In Section~\ref{sec:performance} we continue with performance evaluation of the optimal placement strategies for different probabilistic deployment scenarios. In Section~\ref{discussion} we conclude the paper with discussions.

\section{Model and Problem Definition}
\label{model}
In this section we will present the general model and the problem formulation. 

Throughout the paper we will be interested in different types of base stations, namely MBSs and SBSs with different cache capacities. We will give the most general formulation of the problem as it is possible to have MBSs and SBSs with different storage capacities in some network topologies. 

We consider a heterogeneous cellular network with L-different types of base stations in the plane. These base stations are distributed according to a spatial point process~\cite{baccelli1}. \BS{Type$-\ell$, $\ell = 1,\dots, L$, base stations have coverage radius $r_\ell$}. Let $\mathcal{N}_\ell$ denote the number of base stations of type$-\ell$ that are covering a user at an arbitrary location in the plane. Furthermore, let $p_{\ell}(n_\ell) := \mathbb{P}[\mathcal{N}_\ell = n_\ell]$ denote the probability of a user being covered by $n_\ell$ type$-\ell$ caches, and $p(\bm{n}) := \mathbb{P}[\bm{\mathcal{N}} = \bm{n}]$ \BS{denote the joint probability of a user being covered by $n_\ell$ type$-\ell$ caches, $\forall \ell = 1,\dots,L$}, where $\bm{\mathcal{N}}=(\mathcal{N}_1,\dots, \mathcal{N}_L)$ and $\bm{n} = \left(n_1, \dots, n_L\right).$ We also define  $\bm{\mathcal{N}}^{(-\ell)} = (\mathcal{N}_i)_{i\neq\ell}$ and $\bm{n}^{(-\ell)} = (n_i)_{i\neq\ell}$ as the corresponding $(L-1)$-tuples excluding the $\ell$th component.

Caches store files from a content library $\mathcal{C} = \{c_1, c_2, \dots, c_J\}$, where an element $c_j$ is a file with normalized size $1$. The probability that file $c_j$ is requested is denoted as $a_j$.
Without loss of generality, $a_1 \geq a_2 \geq \dots \geq a_J$. 

Type$-\ell$ caches have cache memory size $K_\ell \geq 1$, $\ell = 1, \dots, L$.  We use the probabilistic content placement policy of~\cite{optimalgeographic}. We denote the probability that the content $c_j$ is stored at a type$-\ell$ cache as
\begin{equation}
b_j^{(\ell)} := \mathbb{P}\left(c_j \text{ stored in  type$-\ell$ cache}\right),
\end{equation}
and the placement strategy $\bm{b}^{(\ell)} = (b_1^{(\ell)}, \dots, b_J^{(\ell)})$  as a J-tuple for any type$-\ell$ cache. The content is independently placed in the cache memories according to the same distribution for the same type of caches. The placement procedure is as follows. The memory of a type$-\ell$ cache is divided into $K_\ell$ continuous memory intervals of unit length. Then $b_j^{(\ell)}$ values fill the cache memory sequentially and continue filling the next slot if not enough space is available in the memory slot that it has started filling in as in the end completely covering the $K_\ell$ memory intervals. Then, for any type$-\ell$ cache, a random number from the interval $[0,1]$ is picked and the intersecting $K_\ell$ files are cached. An example is shown in Figure~\ref{realization}.

Then the overall placement strategy for all types of caches can be denoted by $\bm{B} = \left[\bm{b}^{(1)}; \dots; \bm{b}^{(L)}\right]$ as a $L \times J$ matrix.
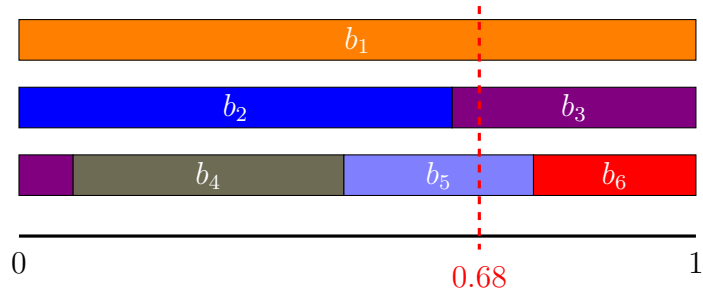
\begin{figure}
\centering
%
  \begin{tikzpicture}[scale=0.18]
\filldraw[draw=black, fill=orange] (0,0) rectangle (50,-3);
\filldraw[draw=black, fill=blue] (0,-5) rectangle (32,-8);
\filldraw[draw=black, fill=blue!50!red] (32,-5) rectangle (50,-8);
\filldraw[draw=black, fill=blue!50!red] (0,-10) rectangle (4,-13);
\filldraw[draw=black, fill=yellow!30!black] (4,-10) rectangle (24,-13);
\filldraw[draw=black, fill=blue!50!white] (24,-10) rectangle (38,-13);
\filldraw[draw=black, fill=red] (38,-10) rectangle (50,-13);
\draw[black,very thick] (0,-16) -- (50,-16);
\node[draw=none,inner sep=2pt,fill=none, label={[black]below:$0$}] at (0,-16) {};
\node[draw=none,inner sep=2pt,fill=none, label={[black]below:$1$}] at (50,-16) {};
\draw[red,very thick, dashed] (34,1) -- (34,-17);
\node[draw=none,inner sep=2pt,fill=none, label={[red]below:$0.68$}] at (34,-17) {};
\node[draw=none,inner sep=2pt,fill=none, label={[white]center: $b_1$}] at (25,-1.5) {};
\node[draw=none,inner sep=2pt,fill=none, label={[white]center: $b_2$}] at (16,-6.5) {};
\node[draw=none,inner sep=2pt,fill=none, label={[white]center: $b_3$}] at (41,-6.5) {};
\node[draw=none,inner sep=2pt,fill=none, label={[white]center: $b_4$}] at (14,-11.5) {};
\node[draw=none,inner sep=2pt,fill=none, label={[white]center: $b_5$}] at (31,-11.5) {};
\node[draw=none,inner sep=2pt,fill=none, label={[white]center: $b_6$}] at (44,-11.5) {};
  \end{tikzpicture}
\caption{A realization of the probabilistic placement policy ($J = 6$ and $K_\ell = 3$). A random number is picked (0.68 in this case) and the vertical line intersects with the optimal $\bm{b}^{(\ell)}$ values. From this figure, we conclude that the content subset $\{c_1, c_3, c_5\}$ will be cached.}
\label{realization}
\end{figure}

Next, we introduce our performance metric. The performance metric is the total miss probability ($1$-minus-hit probability) for the users located in the plane and is given by
\begin{equation}
f\left(\mathbf{B}\right) = \sum_{j = 1}^J a_j \sum_{n_1=0}^{\infty} \dots \sum_{n_L=0}^{\infty}  p(\bm{n}) \prod_{\ell = 1}^L (1 - b_j^{(\ell)})^{n_\ell}.
\label{missprob}
\end{equation}

We define the optimization problem to find the optimal placement strategy minimizing the total miss probability as follows:
\begin{prb}
\label{orgprb}
\begin{align}
&\min \text{ } f\left(\mathbf{B}\right)\nonumber\\
&\text{ }\mathbf{s.t.}\quad  b_1^{(\ell)} + \dots + b_J^{(\ell)} = K_\ell, \quad b_j^{(\ell)} \in [0,1],\quad \forall j, \ell. \label{constraints}
\end{align}
\end{prb}

\section{Deployment models and file popularities}
In this section we present specific cache deployment models and file popularity distributions that will be used in this paper.
\newpage
\subsection{Deployment models}
\subsubsection{Homogeneous PPP deployment model}
\BS{In this model we assume that a user at an arbitrary location in the plane can connect to all type$-\ell$ caches that are within radius $r_\ell$. The caches follow a two-dimensional (2D) spatial homogeneous Poisson process with type$-\ell$ caches independently distributed in the plane with density $\lambda_\ell > 0$ where $\ell = 1, \dots, L$. Type$-\ell$ caches within radius $r_\ell$ follows a Poisson distribution with parameter $t_\ell = \lambda_\ell \pi r_\ell^2$. Then, we conclude that
\begin{align}
p_{\ell}(n_\ell) &= P(\text{$n_\ell$ type$-\ell$ caches within radius $r_\ell$})\nonumber\\
&= \frac{t_\ell^{n_\ell}}{n_\ell!}e^{-t_\ell}. \label{poisdist}
\end{align}

The user is covered by $n_\ell$ type$-\ell$ caches and distributions of the different types of caches are independent of each other. Therefore, the total coverage distribution probability mass function $p_{\bm{n}}$ will be the product of individual probability distributions
\begin{equation}
\label{PPPcoverage}
p\left(\bm{n}\right) = p_{1}(n_1) p_{2}(n_2)\dots p_{L}(n_L).
\end{equation}
}
\subsubsection{M-or-None deployment model}
In this model once again we assume that a user at an arbitrary location in the plane can connect to \BS{all type$-\ell$ caches that are within radius $r_\ell$}. Type$-1$ caches represent macro base stations and follow a two-dimensional (2D) spatial homogeneous Poisson process with density $\lambda_1 > 0$. As a consequence, the number of type$-1$ caches within radius $r_1$ follows a Poisson distribution satisfying~\eqref{poisdist} for $\ell = 1$.

We assume that if a user is covered by at least one macro base station (type$-1$ cache), then it will have $M$ helpers (other types of caches) in total. As a result, network operators serve users with providing them $M$ helpers as long as they are connected to at least one of the macro base stations. If a user is not covered by a type$-1$ cache, then it doesn't receive any service from other caches either. Therefore, we have
\begin{equation*}
P\left(\bm{\mathcal{N}}^{(-1)} = \bm{n}^{(-1)} \vert \mathcal{N}_1 = 0\right) = \left\{
\begin{array}{rl}
0 & \text{if } \sum_{l=2}^L n_l \neq 0,\\
1 & \text{if } \sum_{l=2}^L n_l = 0,
\end{array} \right.
\end{equation*}
and,
\begin{equation*}
P\left(\bm{\mathcal{N}}^{(-1)} = \bm{n}^{(-1)} \vert \mathcal{N}_1 = n_1 \right) = \left\{
\begin{array}{rl}
0 & \text{if } \sum_{l=2}^L n_l \neq M,\\
1 & \text{if } \sum_{l=2}^L n_l = M,
\end{array} \right.
\end{equation*}
when $n_1 > 0$.

\subsection{File Popularities}
In this section we will introduce file popularity distributions. Even though any popularity distribution can be used, our numerical results will be based on Zipf distribution. Specifically we will use standard and perturbed Zipf models.
\subsubsection{Zipf distribution}
For this model, without loss of generality, $a_1 \geq a_2 \geq \dots \geq a_J$. The probability that a user will ask for content $c_j$ is then equal to
\begin{equation}
a_j = \frac{j^{-\gamma}}{\sum_{j=1}^J j^{-\gamma}}, \label{zipfpars}
\end{equation}
where $\gamma > 0$ is the Zipf parameter.

\subsubsection{Perturbed Zipf distribution}
In practice, one might not have the exact file popularities available. Instead, only estimates might be available.
Suppose that $a_j^{\text{pert}}$ values are the actual file popularity values and that $a_j$ values are estimates of these popularities. We propose a perturbed Zipf model for the actual popularity distribution. In this model the probability that a user will ask for content $c_j$ is equal to
\begin{equation}
a_j^{\text{\text{pert}}} = \frac{\left(a_j + Z_j \right)^+}{\sum_{j=1}^J \left(a_j + Z_j \right)^+}, \label{zipfnoisy}
\end{equation}
where $a_j$ follows a Zipf distribution~\eqref{zipfpars} with given $\gamma > 0$, $Z_j$ is the noise, where $Z_j$ is independent and identically distributed and drawn from a zero-mean normal distribution with variance $\sigma_j^2$. Note that the difference between the available popularity values $a_j$ and the actual file popularity values $a_j^{\text{pert}}$ increases as the variance $\sigma_j^2$ of the perturbation increases.

\section{Joint Optimization for the Homogeneous Poisson Point Process (PPP) model}
\label{sec:jointoptimization}

Finding the optimal placement strategy for all types of caches jointly is an interesting problem. However, this joint optimization problem presented in Problem~\ref{orgprb} is not convex in general conditions. 

When each type of cache is deployed according to a homogeneous spatial Poisson process, it is possible to reformulate the joint optimization problem such that the problem becomes convex. We will present such a formulation in the next subsection and continue with the general structure afterwards.

Since each type of cache is deployed according to a homogeneous spatial Poisson process, we can see for any file $j$ being present in a type$-i$ cache as a thinned Poisson process~\cite{stochasticgeometry}. Then, for any user in the plane, the probability of missing file $j$ is equal to the joint probability of the thinned Poisson processes. We will continue with formulation and the optimal solution of this problem.
\subsection{Formulation of the problem}
In this model, we assume that a user at an arbitrary location in the plane can connect to all type$-\ell$ caches that are within radius $r_\ell$. The caches follow a two-dimensional (2D) spatial homogeneous Poisson process with type$-\ell$ caches independently distributed in the plane with density $\lambda_\ell > 0$, where $\ell = 1, \dots, L$. The number of type$-\ell$ caches within radius $r_\ell$ follows a Poisson distribution with parameter $\lambda_\ell \pi r_\ell^2$. Then, from thinning the Poisson process, it follows that type$-\ell$ caches storing the file $c_j$ follows the Poisson distribution with parameter $b_j^{(\ell)}\lambda_\ell \pi r_\ell^2$. The performance metric is the total miss probability which is the probability that a user will miss the content that he requires in one of the caches that he is covered by. 
\begin{lem}
With the thinned Poisson process argument, the total miss probability is given by
\begin{equation}
f_{\text{joint}}\left(\bm{B}\right) = \sum_{j=1}^J a_j \exp\left\{-\sum_{\ell = 1}^L b_j^{(\ell)} \lambda_\ell \pi r_\ell^2\right\}.
\label{thinprob}
\end{equation}
\end{lem}
\begin{proof}
By using the thinning argument, the total probability of miss for type$-\ell$ cache is equal to the probability of being covered by $0$ type$-\ell$ caches storing the file. Therefore, we have
\begin{equation*}
f_{\ell}\left(\bm{b}^{(\ell)}\right) = \sum_{j=1}^J a_j \exp\left\{-b_j^{(\ell)} \lambda_\ell \pi r_\ell^2\right\}.
\end{equation*}
When different types of caches' locations are all following homogeneous Poisson processes, the probability of missing a specific file is equal to the joint probability of being not covered by any cache storing that specific file over all types of caches. Hence, \BS{the union of independent Poisson point processes is a Poisson point process with a density equal to the sum of the respective densities and the total miss probability is given by Eq.~\eqref{thinprob}}.
\end{proof}
We define the optimization problem to find the optimal placement strategy minimizing the total hit probability for all caches jointly as follows:
\begin{prb}
\label{prb:PPPorgprb}
\begin{align}
&\min \text{ } f_{\text{joint}}\left(\bm{B}\right)\nonumber\\
&\text{ }\mathbf{s.t.}\quad b_1^{(\ell)} + \dots + b_J^{(\ell)} = K_\ell, \quad b_j^{(\ell)} \in [0,1],\quad \forall j, \ell. \label{constraints2}
\end{align}
\end{prb}
\subsection{Solution of the optimization problem}
In this section, we will analyze the structure of the optimization problem.
\begin{lem}
\label{convex}
Problem~\ref{prb:PPPorgprb} is a convex optimization problem.
\end{lem}
\begin{proof}
\BSg{The exponential of an affine function is known to be convex. Our performance metric~\eqref{thinprob} is a positively-weighted sum of the exponential of affine combinations, which is still convex.}
\end{proof}
We already showed that $f_{\text{joint}}\left(\bm{B}\right)$ is convex by Lemma~\ref{convex} and the constraint set is linear as given in~\eqref{constraints2}. Thus, the Karush-Kuhn-Tucker (KKT) conditions provide necessary and sufficient conditions for optimality~\cite{KKTref}. We define a new parameter $d_j$ as the sum of the intensities of all thinned Poisson processes for file $c_j$ as follows:
\begin{equation}
d_j = \sum_{\ell = 1}^L b_j^{(\ell)} \lambda_\ell \pi r_\ell^2 \label{djdef},
\end{equation}
and the following vector consisting of the sum of the intensities of all types of caches for all files: $\bm{D} = (d_1, \dots, d_J)$ as a J-tuple.

Then, the total miss probability is given by
\begin{equation}
f_{\text{sum}}\left(\bm{D}\right) = \sum_{j=1}^J a_j \exp(-d_j),
\label{missprob}
\end{equation}
and we have a optimization problem to find the optimal placement strategy minimizing the total hit probability for all caches when caches are following PPP as follows:
\begin{prb}
\label{prb:PPPmodprb}
\begin{align}
&\min \text{ } f_{\text{sum}}\left(\bm{D}\right)\nonumber\\
&\text{ }\mathbf{s.t.}\quad \sum_{j = 1}^J d_j = \sum_{\ell = 1}^L K_\ell \lambda_\ell \pi r_\ell^2,\label{PPPmodconstraints1}\\
&\hspace{1cm}0 \leq d_j \leq \sum_{\ell = 1}^L \lambda_\ell \pi r_\ell^2,\quad \forall j. \label{PPPmodconstraints2}
\end{align}
\end{prb}
Note that~\eqref{PPPmodconstraints1} follows from the combination of the capacity constraint in~\eqref{constraints2} and the definition of the parameter $d_j$ as presented in~\eqref{djdef} (and changing the summation order.). Similarly,~\eqref{PPPmodconstraints2} directly follows from the boundary constraint in~\eqref{constraints2} and the definition of the parameter $d_j$ as presented in~\eqref{djdef}.

Problem~\ref{prb:PPPmodprb} is a nonlinear resource allocation problem and has the same structure as the problem presented in~\cite{nonlinearresource}. As such, although a solution algorithm to give the optimal solution is available, a closed-form expression for this class of problems is not available in general. One of the contributions of this paper is an explicit closed-form solution for $\bm{D}$. Also, we will demonstrate how to find the optimal placement strategies for all types of caches, \ie how to find $\bm{B}$ from $\bm{D}$.

The Lagrangian function corresponding to Problem~\ref{prb:PPPmodprb} is
\begin{align}
&L\left(\bm{d}, \nu, \bm{\eta}, \bm{\omega}\right) = \sum_{j=1}^J a_j \exp({-d_j})+\nu \left(\sum_{j=1}^J d_j - 
\BS{\sum_{\ell = 1}^ L}K_\ell \lambda_\ell \pi r_\ell^2 \right)\nonumber\\
&- \sum_{j=1}^J \eta_j d_j +\sum_{j=1}^J \omega_j \left(d_j - \sum_{\ell = 1}^L \lambda_\ell \pi r_\ell^2 \right), \nonumber
\end{align}
where $\bm{d}$, $\bm{\eta}$, $\bm{\omega} \in \mathbb{R}_+^{J}$ and $\nu \in \mathbb{R}$.

Let $\bar{\bm{d}}$, $\bar{\bm{\eta}}$, $\bar{\bm{\omega}}$ and $\bar{\nu}$ be primal and dual optimal. The KKT conditions for Problem~\ref{prb:PPPmodprb} state that
\begin{align}
\sum_{j=1}^J \bar{d}_j = \sum_{\ell = 1}^L K_\ell \lambda_\ell \pi r_\ell^2&, \label{kkt2ppp}\\
0 \leq \bar{d}_j \leq \sum_{\ell = 1}^L \lambda_\ell \pi r_\ell^2&, \quad \forall j = 1,\dots, J \label{kkt1ppp}\\
\bar{\eta}_j \geq 0&,\quad \forall j = 1,\dots, J, \label{kkt3ppp}\\
\bar{\omega}_j \geq 0&,\quad \forall j = 1,\dots, J, \label{kkt4ppp}\\
\bar{\eta}_j \bar{d}_j = 0&,\quad \forall j = 1,\dots, J, \label{kkt5ppp}\\
\bar{\omega}_j \left(\bar{d}_j - \sum_{\ell = 1}^L \lambda_\ell \pi r_\ell^2\right) = 0&, \quad \forall j = 1,\dots, J, \label{kkt6ppp}\\
-a_j \exp(-\bar{d}_j) + \bar{\nu} - \bar{\eta}_j + \bar{\omega}_j = 0&,\quad \forall j = 1,\dots, J \label{kkt7ppp}.
\end{align}
\begin{thm}
\label{thm:PPPoptsol}
The optimal placement strategy for Problem~\ref{prb:PPPmodprb} satisfies
\begin{align}
\label{PPP:dfunc}
\bar{d}_j = \left\{
\begin{array}{rl}
\sum_{\ell = 1}^L \lambda_\ell \pi r_\ell^2, & \text{if } j < s_1\\
\log{\frac{a_j}{\bar{\nu}}}, & \text{if } s_1 \leq j \leq s_2,\\
0, & \text{if } j > s_2,
\end{array} \right.
\end{align}
where
\begin{align}
\label{PPP:gfunc}
g_j(\nu) = \left\{
\begin{array}{rl}
\sum_{\ell = 1}^L \lambda_\ell \pi r_\ell^2, & \text{if } \nu \leq a_j \exp\left(-\sum_{\ell = 1}^L \lambda_\ell \pi r_\ell^2\right)\\
0, & \text{if } \nu \geq a_j,\\
\log{\frac{a_j}{\nu}}, & \text{otherwise},
\end{array} \right.
\end{align}
and $g: \mathbb{R} \rightarrow \left[0, \sum_{\ell = 1}^L K_\ell \lambda_\ell \pi r_\ell^2\right]$, where $g(\nu) = \sum_{j=1}^J g_j(\nu)$,
\begin{align}
s_1 = \min\biggl\{\,1 \leq k \leq J \mathrel{\Big|} g\left(a_k \exp\left(-\sum_{\ell = 1}^L \lambda_\ell \pi r_\ell^2 \right)\right)  \geq \sum_{\ell = 1}^L K_\ell \lambda_\ell \pi r_\ell^2 \biggr\},\label{PPPs1}
\end{align}
\begin{equation}
s_2 = \max\biggl\{1 \leq k \leq J \mathrel{\Big|} g(a_k) \leq \sum_{\ell = 1}^L K_\ell \lambda_\ell \pi r_\ell^2 \biggr\}, \label{PPPs2}
\end{equation}
and
\begin{equation}
\bar{\nu} = \exp \left\{ \frac{\sum_{j = s_1}^{s_2} \log a_j -\sum_{\ell = 1}^L \lambda_\ell \pi r_\ell^2\left(K_\ell - s_1 + 1\right)}{s_2 - s_1 + 1}\right\}.\label{PPPnubar}
\end{equation}
\end{thm}
\begin{proof}
From~\eqref{kkt5ppp},~\eqref{kkt6ppp} and~\eqref{kkt7ppp}, we have
\begin{equation}
\bar{\omega}_j = \frac{\bar{d}_j}{\sum_{\ell = 1}^L \lambda_\ell \pi r_\ell^2} \left[a_j \exp\left(-\bar{d}_j\right) - \bar{\nu}\right], \label{omegaeqppp}
\end{equation}
which, when insterted into~\eqref{kkt6ppp}, gives
\begin{equation}
\frac{\bar{d}_j}{\sum_{\ell = 1}^L \lambda_\ell \pi r_\ell^2} \left[a_j \exp\left(-\bar{d}_j\right) - \bar{\nu}\right] \left(\bar{d}_j - \sum_{\ell = 1}^L \lambda_\ell \pi r_\ell^2\right) = 0 \label{starppp}.
\end{equation}
From~\eqref{starppp}, we see that $0 < \bar{d}_j < \sum_{\ell = 1}^L \lambda_\ell \pi r_\ell^2$ only if 
\begin{equation*}
\bar{\nu} = a_j \exp\left(-\bar{d}_j\right).
\end{equation*}
Since we know that $0 \leq d_j \leq \sum_{\ell = 1}^L \lambda_\ell \pi r_\ell^2$, this implies that
\begin{align*}
\bar{\nu} \in \left[a_j \exp\left(-\sum_{\ell = 1}^L \lambda_\ell \pi r_\ell^2\right), a_j \right].
\end{align*}
If $\bar{\nu} \leq a_j \exp\left(-\sum_{\ell = 1}^L \lambda_\ell \pi r_\ell^2\right)$, we have $\bar{\omega}_j > 0$.
Thus, from~\eqref{kkt6ppp}, we have $\bar{d}_j = \sum_{\ell = 1}^L \lambda_\ell \pi r_\ell^2$. Similarly, if $\bar{\nu} \geq a_j$, we have $\bar{\eta}_j > 0$. Hence, from~\eqref{kkt5ppp}, we have $\bar{d}_j = 0$.

Recalling $d_j$ from Eq.~\eqref{djdef}, when $\bar{d}_j = \sum_{\ell = 1}^L \lambda_\ell \pi r_\ell^2$, it means that $\bar{b}_j^{(\ell)} = 1$, $\forall \ell = 1,\dots,L$. Similarly, when $\bar{d}_j = 0$, $\bar{b}_j^{(\ell)} = 0$, $\forall \ell = 1,\dots,L$. Then, it follows that there exist $s_1, s_2 \in [1, J]$ such that $\bar{b}^{(\ell)}_1 = \bar{b}^{(\ell)}_2 = \dots = \bar{b}^{(\ell)}_{s_1 -1} = 1$ and $\bar{b}^{(\ell)}_{s_2 + 1} = \bar{b}^{(\ell)}_{s_2 + 2} = \dots = \bar{b}^{(\ell)}_J = 0$, $\forall \ell = 1,\dots,L$. Then $s_1$ is given by
\begin{equation}
s_1 = \min\left\{1 \leq k \leq J \mathrel{\Big|} \bar{\nu} > a_k \exp\left(-\sum_{\ell = 1}^L \lambda_\ell \pi r_\ell^2 \right) \right\}. \label{eqs1earppp}
\end{equation}
In order to satisfy the capacity constraint~\eqref{kkt2ppp}, above minimum is guaranteed to exist. Then $s_1$ is obtained by inserting function $g$ to~\eqref{eqs1earppp} and applying the capacity constraint~\eqref{kkt2ppp}.
\BS{Similarly, }$s_2$ is found by applying the same steps and given by
\begin{equation}
s_2 = \max\left\{1 \leq k \leq J \mathrel{\Big|} \bar{\nu} <a_k \right\}. \label{eqs2earppp}
\end{equation}
Using the same argument, in order to satisfy the capacity constraint~\eqref{kkt2ppp} the above maximum is guaranteed to exist. The proof is completed by solving for $\bar{\nu}$ in $g(\bar{\nu}) = \sum_{\ell = 1}^L K_\ell \lambda_\ell \pi r_\ell^2$.
\end{proof}

\BSg{\begin{thm}
\label{thm:PPPstrategy}
The optimal placement strategy for Problem~\ref{prb:PPPorgprb} satisfies
\begin{align}
\label{PPP:bfunc}
\bar{b}_j^{(\ell)} = \left\{
\begin{array}{rl}
1, \forall \ell, & \text{if } j < s_1\\
\phi_j(\ell), & \text{if } s_1 \leq j \leq s_2,\\
0, \forall \ell, & \text{if } j > s_2,
\end{array} \right.
\end{align}
where $\phi_j(\ell)$ is any solution over $b^{(\ell)}_j\in [0,1]$ satisfying
\begin{equation}
\sum_{\ell = 1}^{L} \bar{b}_j^{(\ell)} \lambda_\ell \pi r_\ell^2 = \log\frac{a_j}{\bar{\nu}}, \label{midconstr}
\end{equation}
\begin{equation}
\bar{b}_{s_1}^{(\ell)} + \dots + \bar{b}_{s_2}^{(\ell)} = K_\ell - s_1 + 1, \label{thm3bconstrst} 
\end{equation}
for all $s_1\leq j\leq s_2$ and $1\leq\ell\leq L$ and where $s_1$, $s_2$ and $\bar{\nu}$ are given by from Eq.~\eqref{PPPs1},~\eqref{PPPs2} and~\eqref{PPPnubar}, respectively.
\end{thm}
\begin{proof}
We will combine the solution given in Theorem~\ref{thm:PPPoptsol} with Eq.~\eqref{djdef} for the proof. In Theorem~\ref{thm:PPPoptsol} we give the solution for the joint problem in terms of $\bm{D}$. This solution is unique in $d_j$'s, and $\bar{d}_j$'s are given in Theorem~\ref{thm:PPPoptsol}. Even though the solution is unique in $d_j$'s, it is easy to observe that the solution is not unique in $b_j$'s from Eq.~\eqref{djdef}. To give an optimal placement strategy in terms of $b_j$'s, we will first show that the relation between the $d_j$'s and $b_j$'s correspond to a \JG{balanced} capacitated transportation problem~\cite{combinatorialoptimization}. Finally, since it is known that the greedy solution is optimal for the balanced capacitated transportation problem, we can obtain $\bar{b}_j^{\ell}$'s greedily.

When $\bar{d}_j = \sum_{\ell = 1}^L \lambda_\ell \pi r_\ell^2$, for any file index $j < s_1$, $\bar{b}_j^{(\ell)} = 1$, $\forall \ell = 1,\dots,L$, \ie file $c_j$ is stored in all types of caches. Similarly, when $\bar{d}_j = 0$, it means that for any file index $j > s_2$, $\bar{b}_j^{(\ell)} = 0$, $\forall \ell = 1,\dots,L$, \ie file \BS{$c_j$ is not stored in any type of caches.} This implies that, the remaining capacity that can be used for files $s_1,\dots,s_2$ in caches of type $\ell$ is $K_\ell - s_1 + 1$. These files should follow~\eqref{midconstr} and~\eqref{thm3bconstrst} and it remains to show that a solution to this system of equations exists.

For notational convenience let $f_j^{(\ell)}=\lambda_\ell\pi r_\ell^2 b_j^{(\ell)}$.
We observe that~\eqref{midconstr} and~\eqref{thm3bconstrst} correspond to a capacitated transportation problem~\cite{combinatorialoptimization} in the variables $f_j^{(\ell)}$, i.e.\ the $f_j^{(\ell)}$ need to satisfy
\begin{align}
\sum_{\ell=1}^L f_j^{(\ell)} &= \log\frac{a_j}{\bar{\nu}}, \label{eq:transp1} \\
\sum_{j=s_1}^{s_2} f_j^{(\ell)} &= \lambda_\ell\pi r_\ell^2(K_\ell - s_1 + 1), \label{eq:transp2} \\ 
0 \leq f_j^{(\ell)} &\leq \lambda_\ell\pi r_\ell^2. \label{eq:transp3} 
\end{align}
In fact, this is a balanced transportation problem, since, by~\eqref{PPPnubar} we have
\begin{align}
\sum_{j = s_1}^{s_2} \log\frac{a_j}{\bar{\nu}}  &= \sum_{j = s_1}^{s_2} \log{a_j} - (s_2 - s_1 + 1)\log{\bar{\nu}} \notag \\
&= \sum_{j = s_1}^{s_2} \log{a_j} - \sum_{j = s_1}^{s_2} \log a_j +\sum_{\ell = 1}^L \lambda_\ell \pi r_\ell^2\left(K_\ell - s_1 + 1\right) \notag \\
&= \sum_{\ell = 1}^L \lambda_\ell \pi r_\ell^2\left(K_\ell - s_1 + 1\right). \label{eq:transp4} 
\end{align}
Moreover, by~\eqref{PPPmodconstraints2} we have
\begin{equation}
\log\frac{a_j}{\bar{\nu}} \leq \sum_{\ell=1}^L \lambda_\ell\pi r_\ell^2. \label{eq:transp5}
\end{equation}
Due to~\eqref{eq:transp5}, the $f_j^{(\ell)}$ can be found for each file satisfying~\eqref{eq:transp1} and~\eqref{eq:transp3}. Finally, it is readily verified that due to~\eqref{eq:transp4} this can be done greedily by considering each file consecutively.
\end{proof}}

\section{Local Optimization}
\label{sec:distributedoptimization}
Since the joint optimization problem presented in Problem~\ref{orgprb} is not convex for general deployment scenarios, \ie if the cache deployments are not following homogeneous PPP, we will provide a heuristic algorithm that finds the optimal solution for a single type of cache assuming that all other types are storing files with fixed probabilities at any iteration step. The main aim is to see how the overall system performance behaves as the algorithm solves for all types of caches iteratively. The procedure is as follows. At each iteration step we find the optimal strategy for a specific type of cache assuming that the placement strategies for other types of caches are known and fixed. Then we continue with the same procedure for the next type and we continue iterating over different types. In the ensuing subsections, first we will formulate the local optimization problem and give the optimal solution for a single type of cache class. Then we continue with presenting our Local Optimization Algorithm by using the local optimization solution we have obtained.
\subsection{Formulation and solution of the problem}
In this section, we will formulate the local optimization problem for a single type of cache where all the other types of caches' placement strategies are known and fixed and give the solution to the problem.

We start this section by defining our performance metric and the formulation of the optimization problem. The performance metric is the total miss probability which is the probability that a user will not find the content that she requires in any of the {\it caches} that she is covered by. We assume that the placement strategy for the probability distribution over $J$ files through all $L-1$ types is fixed and known and we will solve the optimization problem for only one type. In this section, without loss of generality, we consider the optimization of type$-1$.
For notational convenience, superscript $^c$ in the notation $b_j^{(i)^c}$ indicates that the placement strategy for type$-i$ is known and constant. Then, the total miss probability is given by
\begin{equation}
f_{\text{local}}^{(1)}\left(\bm{b}^{(1)}\right) = \sum_{j=1}^J a_j \sum_{n_1= 0}^{\infty}(1-b_j^{(1)})^{n_1}p_{1}(n_1) q_1(j, n_1),
\label{hitprob}
\end{equation}
where
\begin{align}
\label{qfunc}
&q_m(j, n_m) = P(\text{non type$-m$ caches miss the file $c_j$}) \nonumber\\
&= \sum_{\substack{n_k = 0\\k = 1,\dots,L\\k \neq m}}^{\infty} \prod_{\substack{l = 1\\l \neq m}}^L (1-b_j^{(l)^c})^{n_l} P(\bm{\mathcal{N}}^{(-m)} = \bm{n}^{(-m)} \vert \mathcal{N}_m = n_m).
\end{align}

We define the optimization problem to find the optimal placement strategy minimizing the total miss probability for type$-1$ caches as follows:
\begin{prb}
\label{modprb}
\begin{align}
&\min \text{ } f_{\text{local}}^{(1)}\left(\bm{b}^{(1)}\right)\nonumber\\
&\text{ }\mathbf{s.t.}\quad b_1^{(1)} + \dots + b_J^{(1)} = K_1, \quad b_j^{(1)} \in [0,1],\quad \forall j. \label{constraints3}
\end{align}
\end{prb}

Next, we will analyze the structure of the optimization problem.
\begin{lem}
\label{convex2}
Problem~\ref{modprb} is a convex optimization problem.
\end{lem}
\begin{proof}
The objective function is separable with respect to (w.r.t.) $b_1^{(1)}, \dots, b_J^{(1)}$. $f_{\text{local}}^{(1)}\left(\bm{b}^{(1)}\right)$ is convex in $b_j^{(1)}$, $\forall j$. Hence, it is a convex function of $\left(b_1^{(1)}, \dots, b_J^{(1)}\right)$ since it is a sum of convex subfunctions.
\end{proof}
Since $f_{\text{local}}^{(1)}\left(\bm{b}^{(1)}\right)$ is convex by Lemma~\ref{convex2} and the constraint set is linear as given in~\eqref{constraints3}, the Karush-Kuhn-Tucker (KKT) conditions provide necessary and sufficient conditions for optimality. The Lagrangian function corresponding to Problem~\ref{modprb} becomes
\begin{align}
&L\left(\bm{b}^{(1)}, \nu, \bm{\lambda}, \bm{\omega}\right) = \sum_{j=1}^J a_j \sum_{n_1= 0}^{\infty}(1-b_j^{(1)})^{n_1} p_{1}(n_1) q_1(j, n_1) \nonumber\\ 
&+ \nu \left(\sum_{j=1}^J b_j^{(1)} - K_1\right) - \sum_{j=1}^J \lambda_j b_j^{(1)}  + \sum_{j=1}^J \omega_j \left(b_j^{(1)} - 1\right), \nonumber
\end{align}
where $\bm{b}^{(1)}$, $\bm{\lambda}$, $\bm{\omega} \in \mathbb{R}_+^J$ and $\nu \in \mathbb{R}$.

Let $\bar{\bm{b}}^{(1)}$, $\bar{\bm{\lambda}}$, $\bar{\bm{\omega}}$ and $\bar{\nu}$ be primal and dual optimal. The KKT conditions for Problem~\ref{modprb} state that
\begin{align}
\sum_{j=1}^J \bar{b}^{(1)}_j &= K_1, \label{kkt2}\\
0 \leq \bar{b}^{(1)}_j &\leq 1, \quad \forall j = 1,\dots, J, \label{kkt1}\\
\bar{\lambda}_j &\geq 0, \quad \forall j = 1,\dots, J,\label{kkt3}\\
\bar{\omega}_j &\geq 0,\quad \forall j = 1,\dots, J,\label{kkt4}\\
\bar{\lambda}_j \bar{b}^{(1)}_j &= 0,\quad \forall j = 1,\dots, J,\label{kkt5}\\
\bar{\omega}_j \left(\bar{b}^{(1)}_j - 1\right) &= 0, \quad \forall j = 1,\dots, J,\label{kkt6}\\
a_j \sum_{n_1= 0}^{\infty}n_1(1-b_j^{(1)})^{n_1-1}p_{1}(n_1) &q_1(j, n_1) + \bar{\lambda}_j - \bar{\omega}_j = \bar{\nu}, \forall j = 1,\dots, J \label{kkt7}.
\end{align}
\begin{thm}
\label{optsol}
The optimal placement strategy for Problem~\ref{modprb} is
\begin{align}
\label{bfunc}
\bar{b}^{(1)}_j = \left\{
\begin{array}{rl}
1, & \text{if } \bar{\nu} < a_j p_{1}(1) q_1(j, 1)\\
0, & \text{if } \bar{\nu} > a_j \sum_{n_1=0}^{\infty}n_1 p_{1}(n_1) q_1(j, n_1),\\
\phi(\bar{\nu}), & \text{otherwise},
\end{array} \right.
\end{align}
where $\phi(\bar{\nu})$ is the solution over $b^{(1)}_j$ of
\begin{equation}
a_j \sum_{n_1= 0}^{\infty}n_1(1-b_j^{(1)})^{n_1-1} p_{1}(n_1) q_1(j, n_1) = \bar{\nu}, \label{eqsb}
\end{equation}
and $\bar{\nu}$ can be obtained as the unique solution to the additional constraint
\begin{equation}
\bar{b}_1^{(1)} + \dots + \bar{b}_J^{(1)} = K_1. \label{eqK}
\end{equation}
\end{thm}
\begin{proof}
From~\eqref{kkt5},~\eqref{kkt6} and~\eqref{kkt7}, we have
\begin{equation}
\bar{\omega}_j = \bar{b}_j^{(1)} \left[a_j \sum_{n_1=0}^{\infty} n_1 (1-b_j^{(1)})^{n_1 -1} p_{1}(n_1) q_1(j, n_1) - \bar{\nu}\right], \label{omegaeq}
\end{equation}
which, when inserted into~\eqref{kkt6}, gives
\begin{equation}
\bar{b}_j^{(1)}(\bar{b}_j^{(1)} - 1)\left[a_j \sum_{n_1=0}^{\infty}n_1 (1-b_j^{(1)})^{n_1 -1} p_{1}(n_1) q_1(j,n_1)\right] = \bar{\nu} \label{star}.
\end{equation}
From~\eqref{star}, we see that $0 < \bar{b}_j^{(1)} < 1$ only if 
\begin{equation*}
\bar{\nu} = a_j \sum_{n_1=0}^{\infty}n_1 (1-b_j^{(1)})^{n_1 -1} p_{1}(n_1) q_1(j, n_1).
\end{equation*}
Since we know that $0 \leq b_j^{(i)} \leq 1$, this implies that
\begin{align*}
\bar{\nu} \in \left[a_j p_{1}(1) q_1(j, 1),\text{ } a_j \sum_{n_1=0}^{\infty} n_1 p_{1}(n_1) q_1(j,n_1) \right].
\end{align*}
If $\bar{\nu} < a_j p_{1}(1) q_1(j, 1)$, we have
\begin{equation*}
\bar{\omega}_j = \bar{\lambda}_j + a_j \sum_{n_1=0}^{\infty} n_1 (1-b_j^{(1)})^{n_1 -1} p_{1}(n_1) q_1(j, n_1) - \bar{\nu} > 0.
\end{equation*}
Thus, from~\eqref{kkt6}, we have $\bar{b}_j^{(1)} = 1$. Similarly, if $\bar{\nu} > a_j \sum_{n_1=0}^{\infty} n_1 p_{1}(n_1) q_1(j, n_1)$, we have
\begin{equation*}
\bar{\lambda}_j = \bar{\omega}_j +  \bar{\nu} -a_j \sum_{n_1=0}^{\infty}n_1 (1-b_j^{(1)})^{n_1 -1} p_{1}(n_1) q_1(j, n_1) > 0.
\end{equation*}
Hence, from~\eqref{kkt5}, we have $\bar{b}_j^{(1)} = 0$.

Finally, since $\sum_{j=1}^J {b}_j^{(1)}$ is a decreasing function in $\nu$, solving $J$ equations of~\eqref{eqsb} satisfying~\eqref{eqK} give the unique solution $\bar{\nu}$.
\end{proof}

\BSg{\subsection{Local Optimization Algorithm (LOA)}
\label{sec:loa}
The basic idea of our algorithm is to repeatedly perform local optimization. Let ${\tilde{\mathbf{b}}}^{(\ell)}$ denote the optimal the placement strategy for type$-\ell$ caches given by Theorem~\ref{optsol}.

As different types of caches share information with each other, the idea is to see if applying distributed optimization iteratively and updating the file placement strategies over different types of caches gives $\mathbf{\bar{b}}^{(\ell)}$ for all $\ell \in [1:L]$ yielding to the global optimum of Problem~\ref{orgprb}. To check this, we define the following algorithm.

For Local Optimization Algorithm (LOA), we update the caches following the sequence of the indices of the different types of caches. We assume that all types of caches are initially storing the most popular $K_\ell$ files depending on their cache capacities. The algorithm stops when $f^{(\ell)}(\mathbf{b}^{(\ell)})$ converged $\forall \ell \in \{1,\dots,L\}$, \ie a full round over all types of caches $\{1,\dots,L\}$ does not give an improvement in hit probability. LOA is shown in Algorithm~\ref{loa}.
\begin{figure}[!htb]
\removelatexerror
 \begin{algorithm}[H]
\label{loa}
   \caption{Local Optimization Algorithm}
   initialize $\mathbf{b}^{(\ell)} = [\underbrace{1, \dots, 1}_{\text{$K_\ell$ many}}, 0, \dots, 0]$, $\forall \ell \in \{1,\dots,L\}$\;
   set imp = 1\;
   \While {imp = 1}
   {
Set $imp = 0$\;
\For{$\ell = 1:L$}
{
      Solve Problem~\ref{modprb} for type$-\ell$ caches and find $\mathbf{\tilde{b}}^{(\ell)}$ using the information coming from other types of caches\;
      Compute $f^{(\ell)}\left(\mathbf{\tilde{b}^{(\ell)}}\right)$\;
      \If{$f_{\text{local}}^{(\ell)}\left(\mathbf{\tilde{b}^{(\ell)}}\right) - f_{\text{local}}^{(\ell)}\left(\mathbf{b^{(m)}}\right) \neq 0$}
	{
		$imp = 1$
	}
}
}
  \end{algorithm}
\end{figure}}

Next, we will present the placement strategies obtained by LOA for two deployment models we presented earlier.
\BS{
\subsection{LOA for PPP deployment model}
\label{sec:PPPmodel}
We already showed that we can reformulate Problem~\ref{orgprb} and find an analytical solution to the joint optimization problem when all types of caches are deployed according to homogeneous PPP models in Section~\ref{sec:jointoptimization}. In this section, we will follow the local optimization approach and find the optimal solution for different types of caches at every iteration step. The idea behind this is to see if our LOA converges to the optimal solution given in Theorem~\ref{thm:PPPstrategy}. We will give the numerical results in Section~\ref{sec:performance}.

Again, without loss of generality, we consider the local optimization of type$-1$. Since $P(\bm{\mathcal{N}}^{(-1)} = \bm{n}^{(-1)} \vert \mathcal{N}_1 = n_1) = p_{2}(n_2) \dots p_{L}(n_L)$ is independent of $n_1$, for the sake of simplicity we can define a new parameter: $q_1(j) =: q_1(j,n_1)$, $\forall n_1$.

\JG{The proof of the next result follows along the same lines as the proof of Theorem~\ref{optsol} and is skipped due to space constraints.}
\begin{thm}
\label{PPPopt}
LOA solution for of Problem~\ref{modprb} for PPP model is given by
\begin{align}
\label{bfuncPPPLOA}
\tilde{b}^{(1)}_j  = \left\{
\begin{array}{rl}
1, & \text{if } j < s_1,\\
\frac{1}{t_1}\log\frac{a_j t_1 q_1(j)}{\bar{\nu}}, &\text{if } s_1 \leq j \leq s_2,\\
0, & \text{if } j > s_2.
\end{array} \right.
\end{align}
where
\begin{align}
\label{gfuncPPPLOA}
g_j(\nu) = \left\{
\begin{array}{rl}
&1, \quad \text{if } \nu < a_j p_{1}(1) q_1(j),\\
&0, \quad \text{if } \nu > a_j \sum_{n_1=0}^{\infty}n_1 p_{1}(n_1) q_1(j),\\
&\frac{1}{t_1}\log\frac{a_j t_1 q_1(j)}{\nu}, \quad \text{otherwise},
\end{array} \right.
\end{align}
and $g: \mathbb{R} \rightarrow [0, K_1]$, where $g(\nu) = \sum_{j=1}^J g_j(\nu)$,
\begin{equation}
s_1 = \min\left\{1 \leq \ell \leq J \vert g\left(a_{\ell} t_1 e^{-t_1}q_1(\ell)\right) \geq K_1\right\}, \label{eqs1}
\end{equation}
\begin{equation}
s_2 = \max\left\{1 \leq \ell \leq J \middle\vert g\left(a_{\ell} t_1 q_1(\ell)\right) \leq K_1\right\}, \label{eqs2}
\end{equation}
and
\begin{align}
\bar{\nu} = \exp\Biggl\{\frac{\sum_{j = s_1}^{s_2} \log\left(a_j\right) - t_1\left(K_1 - s_1 + 1\right)}{s_2 - s_1 +1}+\log\left[t_1 q_1(j)\right]\Biggr\}. \label{LOAPPPnu}
\end{align}
\end{thm}
}

\subsection{LOA for M-or-None deployment model}
Again, without loss of generality, we first consider the local optimization of type$-1$. For M-or-None deployment model we will first analyze the behavior of the function $q_1(j, n_1)$. Using~\eqref{qfunc}, we have
\begin{equation}
q_1(j,n_1) = \sum_{n_2=0}^{M} \sum_{n_3=0}^{M - n_2} \dots \sum_{n_L = 0}^{M - (n_2 + \dots + n_{L-1})} \prod_{l = 2}^L (1-b_j^{(l)^c})^{n_l},
\end{equation}
which is not a function of $n_1$, but $M$. Then again for the sake of simplicity we define $q_1^M(j) =: q(j,n_1)$ for M-or-None deployment model. The rest of the analysis is the same as the one that is shown for PPP model.
\begin{thm}
\label{MorNoneopt}
LOA solution of Problem~\ref{modprb} for type$-1$ caches for M-or-None model is given by
\begin{align}
\label{bfuncMorNoneLOA}
\tilde{b}^{(1)}_j  = \left\{
\begin{array}{rl}
1, & \text{if } j < s_1\\
\frac{1}{t_1}\log\frac{a_j t_1 q_1^M(j)}{\bar{\nu}}, &\text{if } s_1 \leq j \leq s_2,\\
0, & \text{if } j > s_2.
\end{array} \right.
\end{align}
where
\begin{align}
\label{gfuncMorNoneLOA}
g_j(\nu) = \left\{
\begin{array}{rl}
&1, \quad \text{if } \nu < a_j p_{1}(1) q_1^M(j),\\
&0, \quad \text{if } \nu > a_j \sum_{n_1=0}^{\infty}n_1 p_{1}(n_1) q_1^M(j),\\
&\frac{1}{t_1}\log\frac{a_j t_1 q_1^M(j)}{\nu}, \quad \text{otherwise},
\end{array} \right.
\end{align}
and $g: \mathbb{R} \rightarrow [0, K_1]$, where $g(\nu) = \sum_{j=1}^J g_j(\nu)$,
\begin{equation}
s_1 = \min\left\{1 \leq \ell \leq J \vert g\left(a_{\ell} t_1 e^{-t_1}q_1^M(\ell)\right) \geq K_1\right\}, \label{eqs1mon}
\end{equation}
\begin{equation}
s_2 = \max\left\{1 \leq \ell \leq J \vert g\left(a_{\ell} t_1 q_1^M(\ell)\right) \leq K_1\right\}, \label{eqs2mon}
\end{equation}
and
\begin{align}
\bar{\nu} = \exp\Biggl\{\frac{\sum_{j = s_1}^{s_2} \log\left(a_j\right) - t_1\left(K_1 - s_1 + 1\right)}{s_2 - s_1 +1}+\log\left[t_1 q_1^M(j)\right]\Biggr\}. \label{LOAMorNonenu}
\end{align}
\end{thm}
\begin{proof}
Proof is the same as of PPP model with replacing $q_1(j)$ with $q_1^M(j)$. The only important thing here to note is that $q_1(j,n_1)$ is not a function of $n_1$ for both deployment models and constant for $M$. Thus, solution to~\eqref{eqsb} can be further exploited with some manipulations.
\end{proof}

For the helpers, the analysis is different. Again, without loss of generality, we will consider the local optimization of type$-2$ helpers. For M-or-None deployment model we will first analyze the behavior of the function $q_2(j, n_2)$, \ie the probability that non type$-2$ caches are missing file $j$. First for the sake of simplicity, we define a new parameter for the probability of other helpers missing file $j$ as
\begin{equation*}
\zeta_2(j,n_2) = \sum_{n_3=0}^{M - n_2} \dots \sum_{n_L = 0}^{M - \sum_{k = 2}^{L-1} n_k} \prod_{l = 3}^L \left(1-b_j^{(l)^c}\right)^{n_l}.
\end{equation*}
Then, using~\eqref{qfunc}, we have
\begin{align*}
q_2(j,n_2) &= \sum_{n_1 = 0}^\infty p_1(n_1) \left(1-b_j^{(1)^c}\right)^{n_1} \zeta_2(j,n_2)\\
&= e^{\left(1-b_j^{(1)^c}\right)} \zeta_2(j,n_2).
\end{align*}
$q_2(j,n_2) $ is now a function of $n_2$ and we can not manipulate Eq.~\eqref{eqsb} further to get a closed form solution for the helpers for M-or-None deployment model.

\JG{The proof of the next result follows along the same lines as the proofs of Theorems~\ref{optsol} and~\ref{MorNoneopt}  and is skipped due to space constraints.}
\begin{thm}
\label{MorNoneopt2}
LOA solution of Problem~\ref{modprb} for type$-2$ caches for M-or-None model is given by
\begin{align}
\label{bfuncMorNone2}
\tilde{b}^{(2)}_j = \left\{
\begin{array}{rl}
1, & \text{if } \bar{\nu} < a_j e^{\left(1-b_j^{(1)^c}\right)} \zeta_2(j,1)\\
0, & \text{if } \bar{\nu} > a_j \sum_{n_2=0}^{M}n_2 e^{\left(1-b_j^{(1)^c}\right)} \zeta_2(j,n_2),\\
\phi(\bar{\nu}), & \text{otherwise},
\end{array} \right.
\end{align}
where $\phi(\bar{\nu})$ is the solution over $b^{(2)}_j$ of
\begin{equation}
a_j \sum_{n_2 = 0}^{M}n_2(1-b_j^{(2)})^{n_2-1} e^{\left(1-b_j^{(1)^c}\right)} \zeta_2(j,n_2) = \bar{\nu},
\end{equation}
and $\bar{\nu}$ can be obtained as the unique solution to the additional constraint
\begin{equation}
\bar{b}_1^{(2)} + \dots + \bar{b}_J^{(2)} = K_2. 
\end{equation}
\end{thm}

The solution for other helper types can be obtained by following the same procedure by replacing $\zeta_2(j,n_2)$ by $\zeta_\ell(j,n_\ell)$ for type$-\ell$ caches.

\section{Performance Evaluation}
\label{sec:performance}
In this section, first we will present some heuristic placement strategies. Next we will specify different network coverage models and show the performances of the proposed algorithms .

\subsection{Heuristics}
In this subsection we will introduce some heuristic placement strategies. The main aim of proposing these heuristics is to compare the hit probability performance of the system when the optimal strategy is used with the hit probability obtained when these heuristics are used. Later we will show by numerical results that the hit probability is increased remarkably by using the placement strategies that our proposed algorithms give compared to these heuristics.
\subsubsection{Heuristic 1 (H1)}
\label{policyh1}
The first heuristic is to use is to store the first $K_i$ most popular files in type-$i$ caches, denoted by H1. For H1, 
$$
\bm{{b}}^{(i)} = \left(\underbrace{1, 1, \dots, 1}_\text{$K_i$ many}, 0, \dots, 0\right).
$$ 

\subsubsection{Heuristic 2 (H2)}
\label{policyh2}
We will introduce an example to explain how H2 works. In some scenarios type-$1$ caches may store the first $K_1$ files with high probabilities. Then, it is wiser to come up with a smarter heuristic than H1 since the first $K_1$ files are already available for the users covered by type$-1$ caches. Hence, the second heuristic we propose suggests not to store the most popular first $K_1$ files in type-$2$ caches, and store the next $K_2$ popular files with probability $1$, and continue with the same procedure for type-$3$ caches and so on. The second heuristic is called Heuristic 2 (H2). For H2, 
$$
\bm{{b}}^{(2)} = \left(\underbrace{0, 0, \dots, 0}_\text{$K_1$ many}, \underbrace{1, \dots, 1}_\text{$K_2$ many}, 0, \dots, 0\right),
$$
and 
$$
\bm{{b}}^{(3)} = \left(\underbrace{0, 0, \dots, 0}_\text{$K_1 + K_2$ many}, \underbrace{1, \dots, 1}_\text{$K_3$ many}, 0, \dots, 0\right),
$$
and so on.
\subsubsection{Heuristic 3 (H3)}
\label{policyh3}
We will introduce a smarter deployment heuristic here that also takes the deployment densities of the different types of caches into account. Suppose there are type-$1$ caches in the plane with density $\lambda_1$ and type-$2$ caches are to be deployed in the plane with density $\lambda_2$. Then, we store the first $K_2\lceil \frac{\lambda_2\ r_2^2}{\lambda_1 r_1^2}\rceil$ files with probability $\frac{1}{\lceil \frac{\lambda_2 r_2^2}{\lambda_1 r_1^2}\rceil}$. Namely, for H3,
$$
\bm{{b}}^{(2)} = \left(\underbrace{\frac{1}{\lceil \frac{\lambda_2 r_2^2}{\lambda_1 r_1^2}\rceil}, \frac{1}{\lceil \frac{\lambda_2 r_2^2}{\lambda_1 r_1^2}\rceil}, \dots, \frac{1}{\lceil \frac{\lambda_2 r_2^2}{\lambda_1 r_1^2}\rceil}}_\text{$K_2\lceil \frac{\lambda_2 r_2^2}{\lambda_1 r_1^2}\rceil$ many}, 0, \dots, 0\right). 
$$

\subsection{Poisson Point Process (PPP) deployment model}

\BSg{In this subsection we will consider various scenarios for the case where different types of caches are all following homogeneous PPP. First, we will show the hit probability evolution for the case where helpers (SBSs) with different coverage radii (Femtocells have a typical coverage radius of $10$ m, picocells have $150$ m, and macrocells have $1-2$ km in rural areas~\cite{mobnetguide}.) and different cache capacities. We will illustrate how optimal placement policies behave for LOA, and compare LOA performance with the joint solution and the heuristics. Furthermore, we will show that LOA indeed gives the optimal solution by comparing it with the solution of the joint problem which has been proven to be optimal for homogeneous PPP. 
Also, we will provide numerical results for the case where the file popularities follow distributions with different Zipf parameters and where there is incomplete information on file popularities.}

\subsubsection{Files with popularities following the Zipf distribution}

\BS{
First, we will present a scenario where the SBSs have different coverage radius, $r_2$. Consider the case of two types of caches in the plane. Type$-1$ caches represent MBSs and type$-2$ caches represent SBSs, with $K_1$ and $K_2$-slot cache memories, respectively. The content library size is $J = 1000$. We set $K_1 = 10$ (1\% of the total library) and $K_2$ will have different values, \ie 10, 20, 50, 100 (1, 2, 5 and 10\% of the total library size, respectively.). We set the Zipf parameter $\gamma = 1$ and taking $a_j$ according to~\eqref{zipfpars}. Also we set $\lambda_1 = 1.8324 \times 10^{-5}$, $r_1 = 700$ m and $\lambda_2 = 2\lambda_1$. As a side note, we have chosen the intensity of the Poisson process of the MBSs (or type$-1$ caches) equal to the density of base stations in the real network considered in~\cite{openmobilenetwork, diststorage, sigmetrics}.

\BSg{In \JG{Figure~\ref{hitprobevoradii}} we see that the hit probability increases as the coverage radius for the SBSs increases. If the SBSs are femtocells ($r_2 = 10$ m), increasing the cache capacity will not give a significant improvement in hit probability. However, for picocells ($r_2 = 150$ m), having a larger cache memory significantly increases the hit probability.}

\begin{figure}
\centering
\subfloat[\label{hitprobevoradii}]{
\includegraphics[width=0.45\linewidth]{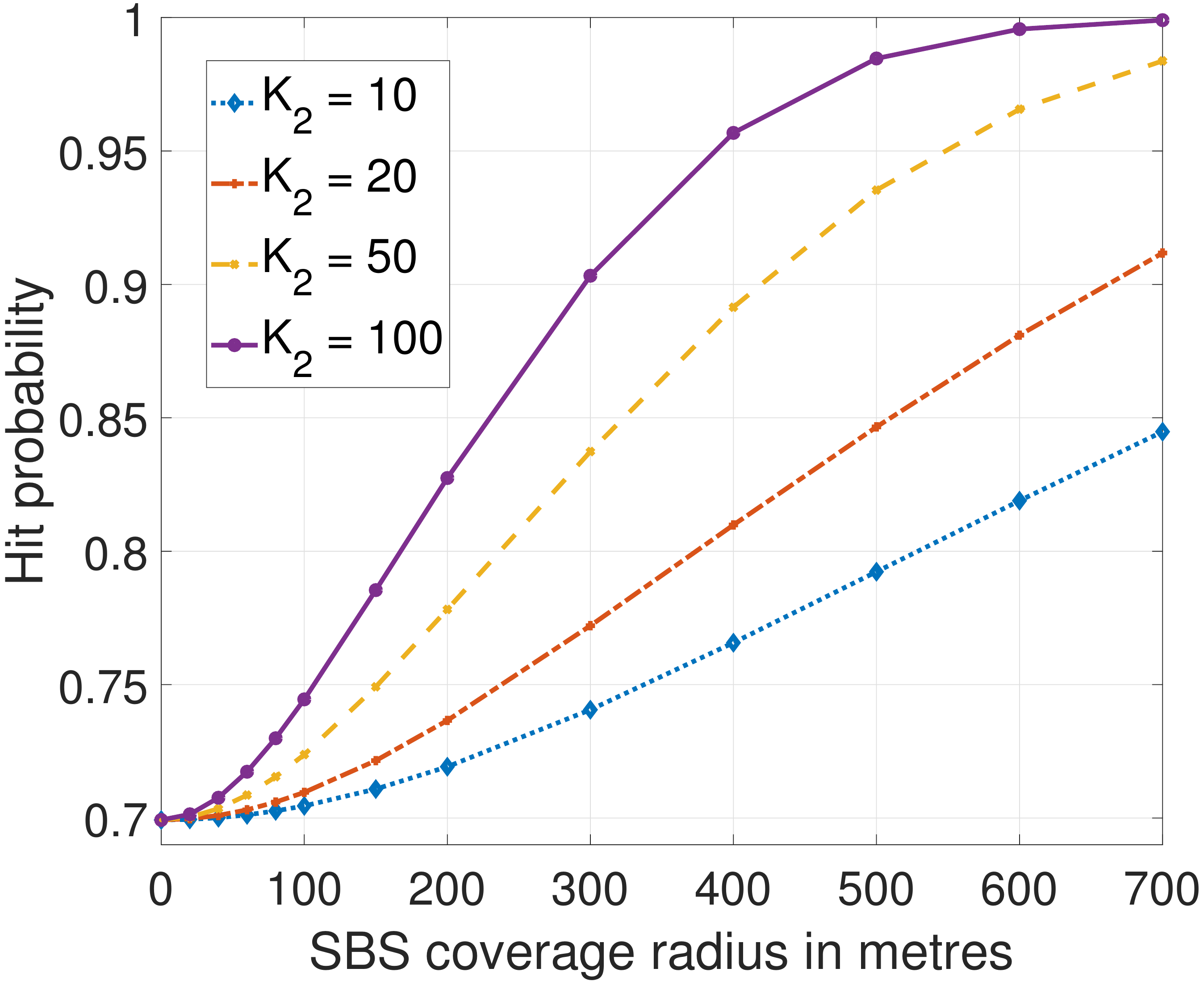}}
\hspace{0.1cm}
\subfloat[\label{hitprobevo}]{
\includegraphics[width=0.42\linewidth]{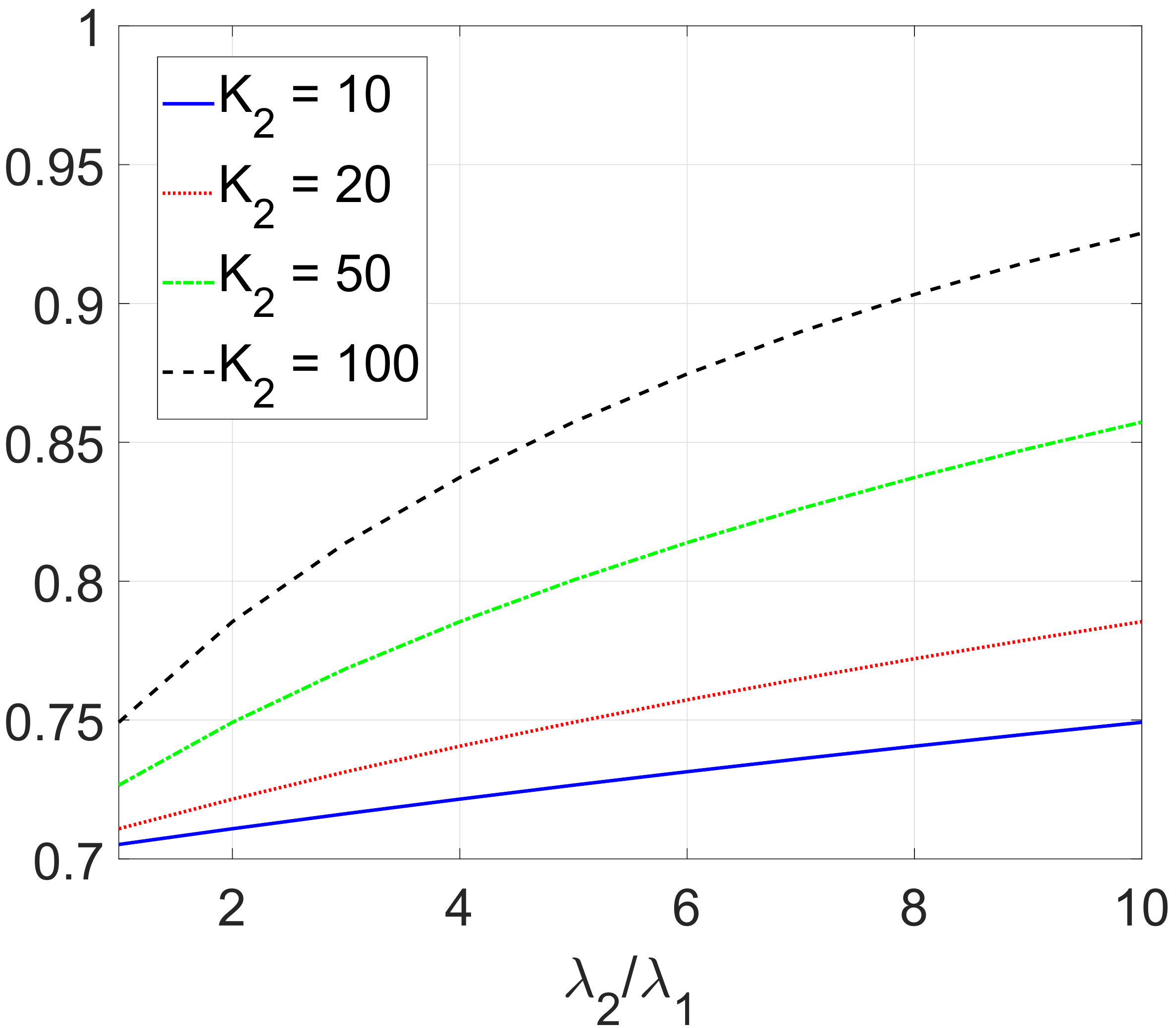}}
\caption{(a) The total hit probability evolution for different SBS coverage radii. (b) The total hit probability evolution for different SBS deployment densities ($r_2 = 150$ m).}
\end{figure}

\BSg{Next, we will present a scenario where the SBSs have a fixed coverage radius and different deployment densities with different cache memories. We assume that our SBSs are picocells. We consider the same parameters in the previous setting, except we set $r_2 = 150 m$ and we will consider various values for the deployment densities for SBSs, namely for $\lambda_2$.}

\BSg{In \JG{Figure~\ref{hitprobevo}} we see the hit probability evolution for the different deployment densities of the SBSs with different cache capacities. Let us consider the two curves where $K_2 = 10$ and $K_2 = 20$. It can be easily verified that the hit probability is equal for the cases when the SBSs have cache memory of $K_2 = 10$ at $\lambda_2/\lambda_1 = 2$ and of $K_2 = 20$  at $\lambda_2/\lambda_1 = 1$. Note that this holds for any ratio, and confirms the validity of the analytical results (for instance, the relation can be seen in Eq.~\eqref{PPPmodconstraints1}.).}

\BSg{Next, we will show how LOA works for the PPP model. Since we can already obtain the optimal solution for the joint problem for this model, we know what the optimal hit probability is. Therefore, we would like to give an insight on how LOA works and performs by comparing it with the joint solution. We consider the same parameters in the previous setting, except we set the content library size $J = 100$ in order to effectively show the difference between the optimal placement probabilities obtained via the joint solution and LOA.}

\begin{figure}
\centering
\includegraphics[width=0.6\columnwidth]{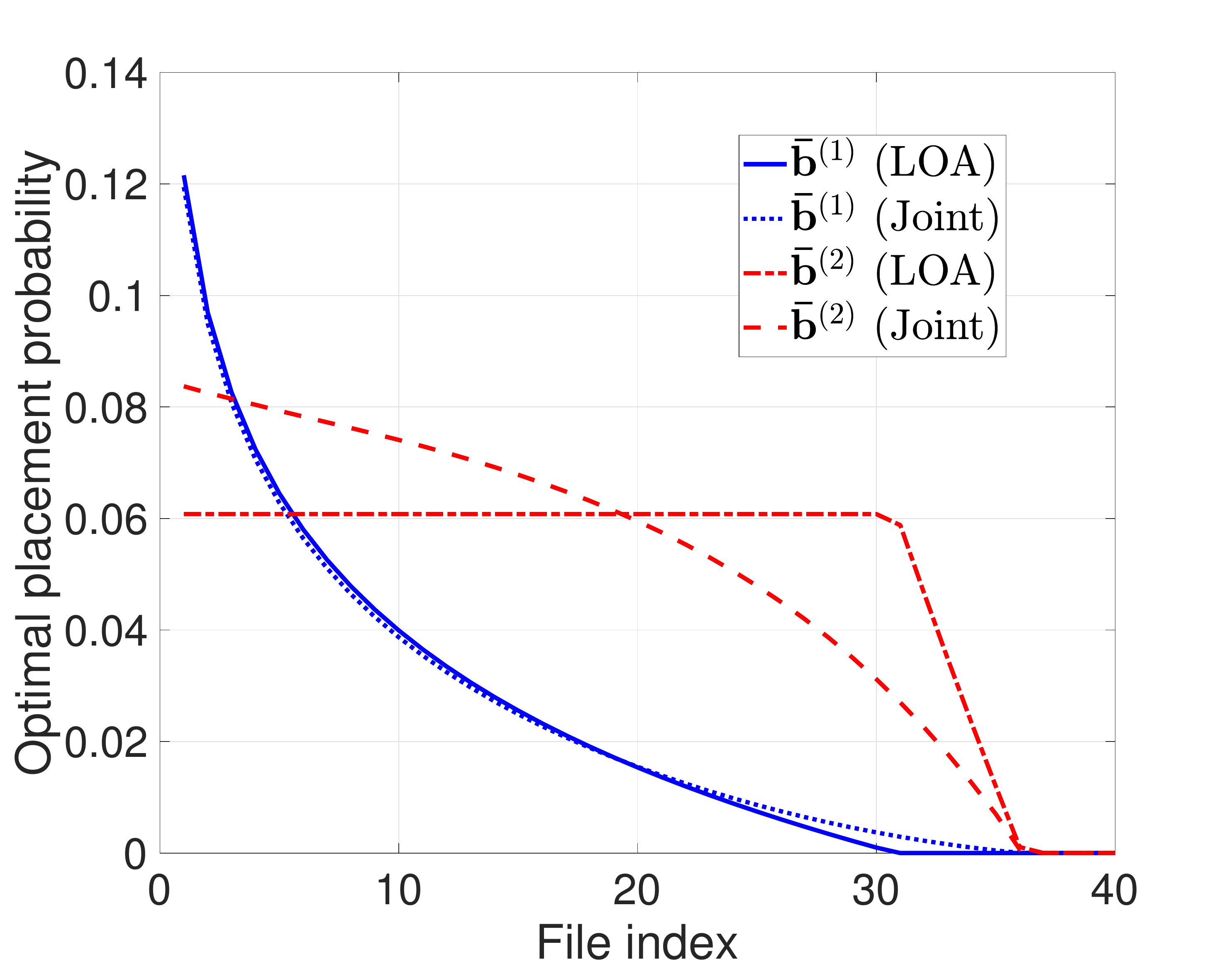}
\caption{Resulting placement strategies obtained via the joint solution and LOA.}
\label{JointVsLoaPlacement}
\end{figure}

In Figure~\ref{JointVsLoaPlacement} we see that the resulting placement strategies for MBSs and SBSs differ for the solution of the joint problem and LOA. For LOA solution, first we find the optimal solution for MBSs assuming that there are no SBSs in the plane. Solving this optimization problem gives the blue straight curve $\bm{\bar{b}}^{(1)}$ (LOA) in Figure~\ref{JointVsLoaPlacement}. Then we set $\bm{b}^{(1)^c} = \bm{\bar{b}}^{(1)}$ (LOA), take it as an input, add SBSs into the plane and find the optimal placement strategy $\bm{\bar{b}}^{(2)}$ (LOA) for SBSs. Solving this optimization problem gives the red dashed curve $\bm{\bar{b}}^{(2)}$ (LOA) in Figure~\ref{JointVsLoaPlacement}. The solution of the joint problem gives the blue dotted curve $\bm{\bar{b}}^{(1)}$ (Joint) and the red dashed dotted red curve $\bm{\bar{b}}^{(2)}$ (Joint), for the MBSs and SBSs, respectively. Computing the hit probabilities for both the joint solution and LOA gives $ f\left(\bm{\bar{b}} \text{ (Joint)}\right) = f\left(\bm{\bar{b}} \text{ (LOA)}\right) = 0.6125$.

\BSg{Now let us pick a file index and verify the validity of Theorem~\ref{thm:PPPstrategy} numerically. Suppose we pick the file $c_1$. $\bar{b}_1^{(1)} \text{ (Joint)} = 0.1195$,  $\bar{b}_1^{(1)} \text{ (LOA)} = 0.1216$, $\bar{b}_1^{(2)} \text{ (Joint)} = 0.0837$ and $\bar{b}_1^{(2)} \text{ (LOA)} = 0.0608$. It is easy to verify that \eqref{midconstr} holds both for the joint problem and LOA. In fact, \eqref{midconstr} holds for any file index and we conclude that LOA gives the optimal solution for this specific scenario even though LOA's resulting placement probabilities are different than the optimal solution of the joint problem. 

This validates our claim in Theorem~\ref{thm:PPPstrategy} numerically, namely $d_j$'s given in Theorem~\ref{thm:PPPoptsol} is unique, and $b_j$'s can be obtained greedily by using LOA, satisfying the unique solution given by Theorem~\ref{thm:PPPstrategy}. 

We would like to briefly give the intuition behind LOA. The algorithm starts with optimizing type$-1$ caches. This solution has already proven to be optimal when there are no other types of caches in the network. On the other hand, adding other types of caches can only help type$-1$ caches, and can not cause any harm to the optimal strategy that had already been obtained. The structure of Theorem~\ref{thm:PPPstrategy} shows that the contribution for the coverage per file per any type of cache comes with two important parameters: (a) the density of the caches and (b) the probability of storing the file. The cumulative contribution of different types of caches will then give the ultimate hit probability. Therefore, just as in Theorem~\ref{PPPopt}, starting with the worst case strategy (no information coming from other types for type$-1$) and updating all types of caches sequentially by providing more information to the next type at each step is simply equivalent to the optimal strategy for the joint problem as provided in~Theorem~\ref{thm:PPPstrategy}. As a conclusion, LOA is a distributed algorithm that gives the optimal placement strategy when caches are deployed according to PPP.
 }

\JG{Numerical results  for different Zipf parameters are consistent with the presented results.
}

Our next aim is to compare the optimal placement strategy with various heuristics. We use the same simulation parameters as in the above cases with the content library size $J = 100$.

\begin{figure}
\centering
\includegraphics[width=0.6\columnwidth]{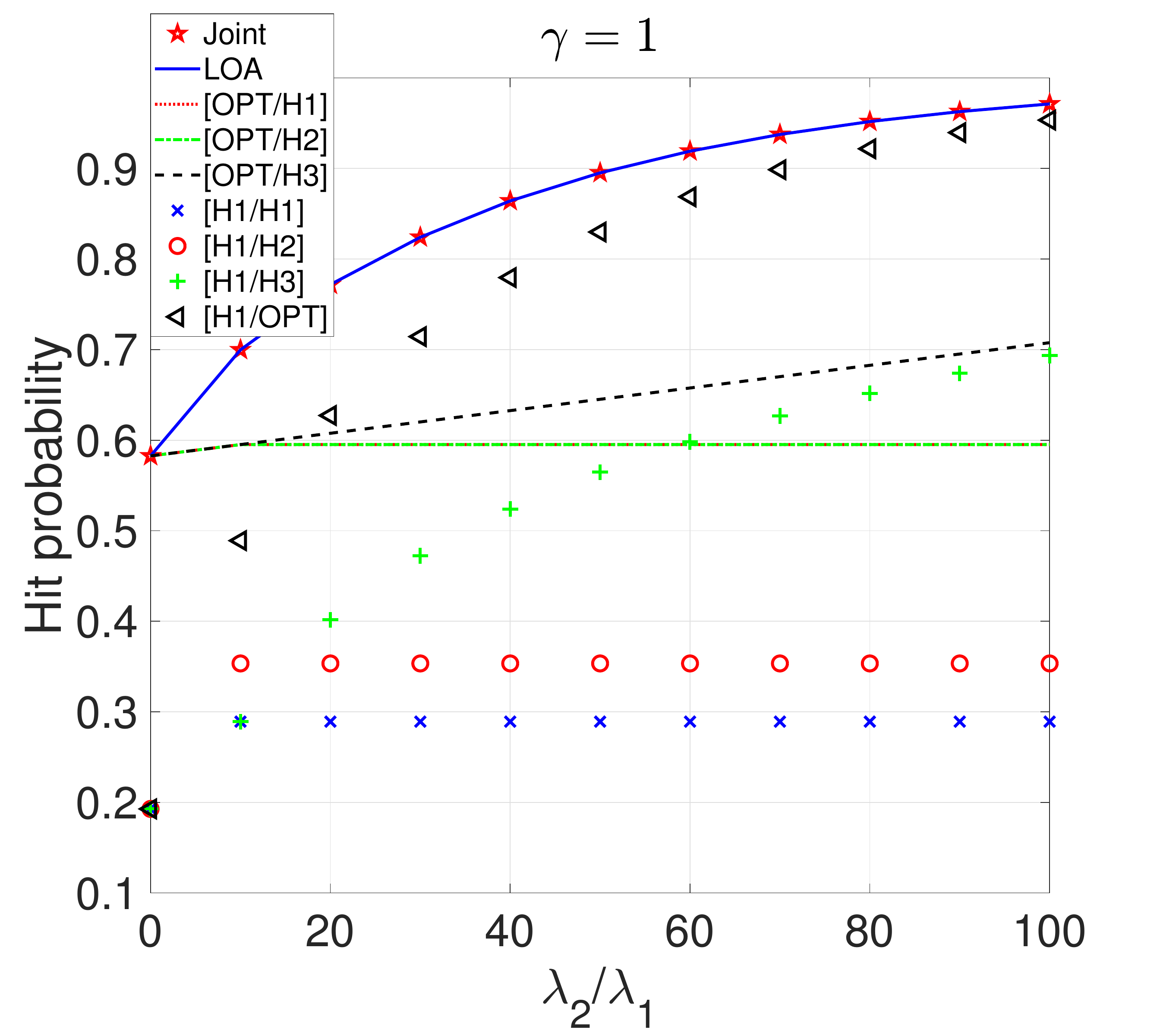}
\caption{Hit probability evolution for various heuristics.}
\label{gamma1}
\end{figure}

\BSg{In Figure~\ref{gamma1}, we see the hit probability evolution for various heuristics. We observe that:
\begin{itemize}
\item LOA is optimal since it performs equally well as the joint solution.
\item As $\lambda_2 \pi r_2^2 >> \lambda_1 \pi r_1^2$, using H3 for the deployment of SBSs gives an improvement in hit probability as long as MBSs are using the optimal placement strategy [OPT] -which is simply the first iteration step of LOA, \ie finding the optimal solution for MBSs assuming that there are no SBSs in the plane.-. However, it requires an unrealistically dense deployment of the SBSs in order to reach the optimal performance.
\item Applying first step of LOA to MBSs and using H1 and H2 for SBSs results in significant performance penalties.
\item Using H1 for the MBSs and applying the second step of LOA to SBSs, \ie taking H1 strategy for the MBSs fixed, and finding the optimal strategy for SBSs, performs significantly well and converges to optimal performance as $\lambda_2 \pi r_2^2 >> \lambda_1 \pi r_1^2$.
\item Applying heuristics to both MBSs and SBSs results in significant performance penalties.
\item Running LOA iteratively does not improve the hit probability, namely running one round over each type of caches gives the optimal solution.
\end{itemize}}

\BSg{\subsubsection{Incomplete information on file popularities}

We use the same notation in the previous subsection. The content library size is $J = 100$. We set $K_1 = 1$ and $K_2 = 2$. We set $\gamma = 1$, $a_j$ takes the values from~\eqref{zipfpars} and $a_j^{\text{pert}}$ takes the values from~\eqref{zipfnoisy} by adjusting $\sigma_j^2$ such that the signal-to-noise ratio ($SNR$) between $a_j$ and $\sigma_j^2$ is set accordingly, \ie the signal power is equal to $POW_{a_j} = 10\log_{10}\left(\vert a_j \vert ^2\right)$, and the noise power is equal to $POW_{\sigma_j^2} = POW_{a_j} - SNR$ in dB. We set $\lambda_1 = 0.5$ and $r_1 = r_2 = 1$.

\begin{figure}
\centering
\input{figures/Figure6.tex}
\label{PPPSmartPerturb}
\end{figure}

As proposed earlier, in real life most of the time the file popularities will not follow a smooth distribution as the Zipf distribution. Recalling from the model that $a_j^{\text{pert}}$ values are the actual file popularity values that can not be obtained in real-time. We have the approximated $a_j$ values available and difference between the available popularity values $a_j$ and the actual file popularity values $a_j^{\text{pert}}$ increases as $\sigma_j^2$ increases. In Figure~\ref{PPPSmartPerturb} we show the total hit probability evolution for LOA. Straight lines indicate the ideal maximum hit probability that could be reached if the optimal deployment strategy was found by using $a_j^{\text{pert}}$ values. Dashed lines show the hit probability when the system is optimized with the already available $a_j$ values. It is \JG{not surprising} that the difference between the ideal and actual hit probability decreases as $\sigma_j^2$ decreases. We see a similar behavior under the [H1/OPT] strategy in \JG{our simulations.}
 }

\BSg{\subsection{M-or-None deployment model}

In this subsection we will present the performance evaluation of the placement strategies for the files following the Zipf distribution for M-or-None deployment model.

\subsubsection{Files with popularities following the Zipf distribution}
We use the same notation in the previous subsection. The content library size is $J = 100$. We set $K_1 = 1$, and $K_2 = 5$. We set $\gamma = 1$ and taking $a_j$ according to~\eqref{zipfpars}. Also we set $\lambda_1 = 1.8324 \times 10^{-5}$ and $r_1 = 700 m$.

For the first step of LOA, the optimal placement strategy for MBSs is $\bm{\bar{b}}^{(1)} = \left(0.1220, 0.0973,\right.\\\left.0.0829, 0.0727, \dots \right)$, and the resulting hit probability is $ f\left(\bm{\bar{b}}^{(1)}\right) = 0.5875$. Next, we solve the problem for Type-2 caches.
\begin{figure}
\centering
\subfloat[\label{cachingpolicyMorNoneSmart}]{
\includegraphics[width=0.45\linewidth]{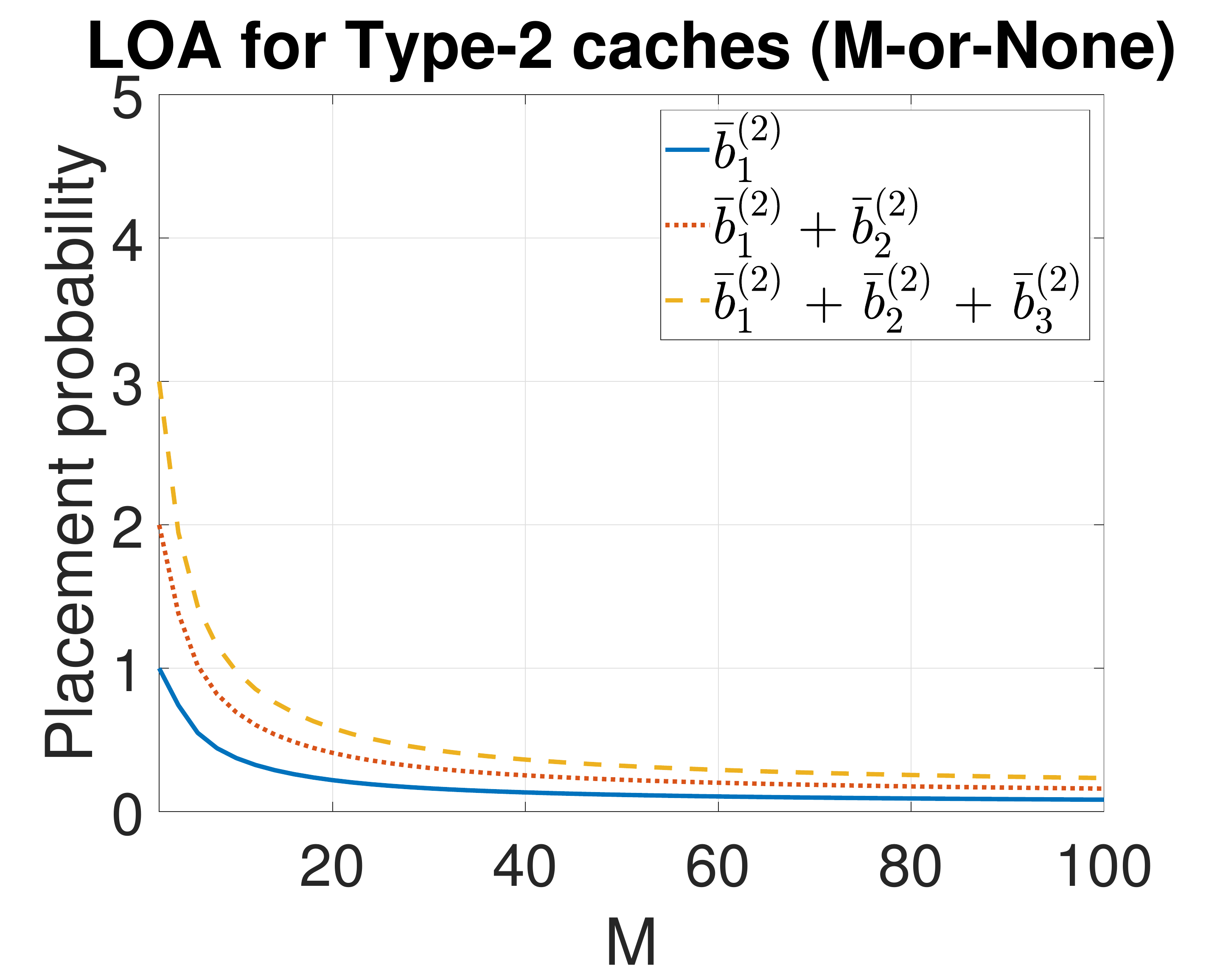}}
\hspace{0.1cm}
\subfloat[\label{cachingpolicyMorNoneMostpop}]{
\includegraphics[width=0.43\linewidth]{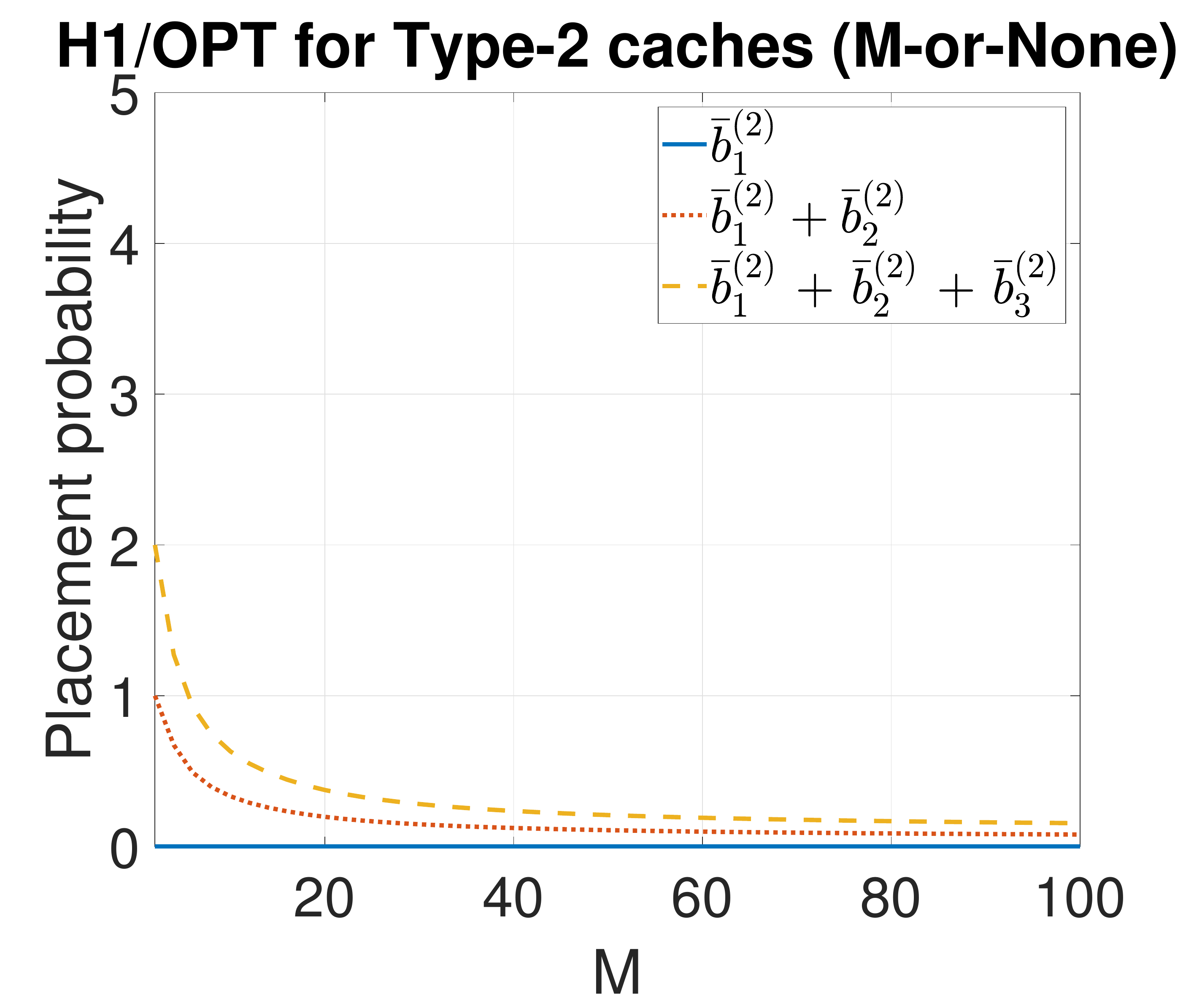}}
\caption{(a) Optimal placement strategy $\bm{\bar{b}}^{(2)}$ for SBSs for different $M$ values [LOA] (M-or-None). (b) Optimal placement strategy $\bm{\bar{b}}^{(2)}$ for SBSs for different $M$ values [H1/OPT] (M-or-None).}
\end{figure}
From Figure~\ref{cachingpolicyMorNoneSmart}, we see that the probability of storing less popular files increases as $M$ increases. As in the PPP case, repeatedly updating $\bm{\bar{b}}^{(1)}$ and $\bm{\bar{b}}^{(2)}$, does not improve the hit probability. We could not come up with an analytical solution to the joint problem of this deployment model since it is not convex, however from the result we obtained from the PPP model, it is very likely that LOA algorithm performs quite well for M-or-None deployment model as well.

For the [H1/OPT] scenario, we have $\bm{\bar{b}}^{(1)} = \left(1, 0, \dots, 0\right)$ and the resulting hit probability is $ f\left(\bm{\bar{b}}^{(1)}\right) = 0.2040$. Then we set $\bm{b}^{(1)^c} = \bm{\bar{b}}^{(1)}$, take it as a fixed input, and find the optimal placement strategy $\bm{\bar{b}}^{(2)}$ for SBSs.

\begin{figure}
\centering
\includegraphics[width=0.6\columnwidth]{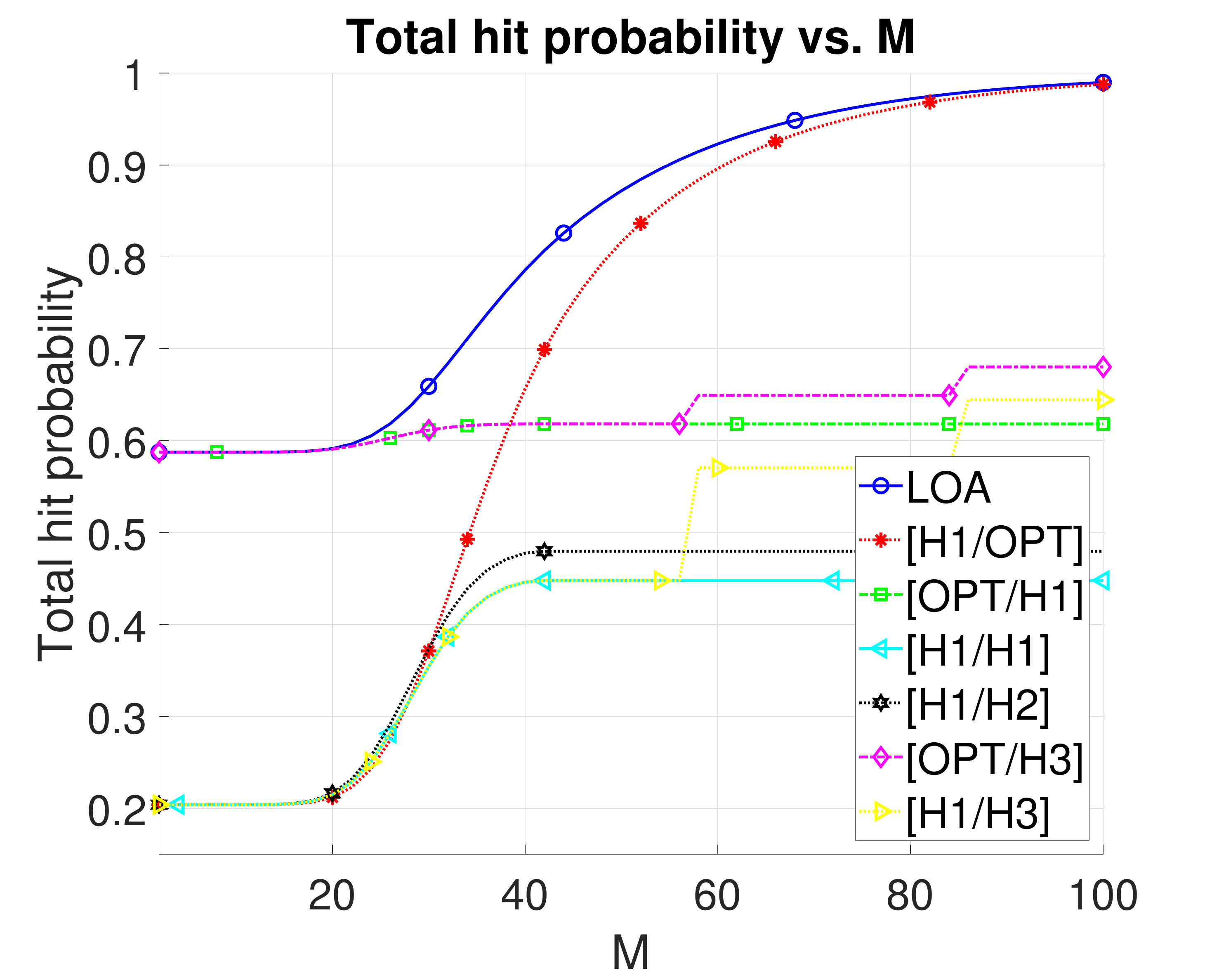}
\caption{The total hit probability evolution for different $M$ values (M-or-None).}
\label{hitprobMorNone}
\end{figure}
From Figure~\ref{cachingpolicyMorNoneMostpop}, as $c_1$ is stored in MBSs with probability $1$ and SBSs are present in the system only when $n_1 > 0$, $c_1$ is never stored at SBSs. We see that probability of storing $c_2$ and $c_3$ decreases and probability of storing other files increases as $M$ increases.

From Figure~\ref{hitprobMorNone}, we see that the hit probabilities under LOA and [H1/OPT] become identical and increase as $M$ increases, \ie we can use a heuristic placement policy for MBSs as long as we compensate the penalty by optimizing SBSs. The hit probability remains constant for until some point as $M$ increases due to the nature of the M-or-None model. Until the point where the hit probability starts increasing, only type-1 caches are present in the system (note that $\lambda_1 \pi r_1^2 \approx 28.2$). For heuristic SBS deployment policies, [OPT/H1], [H1/H1] and [H1/H2] policies achieve significantly lower hit probability than the optimal policy. We use a small variant of H3 here (since we do not have the density ratios for H3, we used the ratio of $\frac{M - \lambda_1 \pi r_1^2}{\lambda_1 \pi r_1^2}$.) [OPT/H3] and [H1/H3] achieve a higher probability compared to other heuristics, however optimizing the SBSs still gives a much higher hit probability.}

\BSg{We see a similar behavior in M-or-None deployment model for the incomplete information on file popularities as in PPP case, and we skip the illustration due to space constraints.}

}

\section{Discussion and Conclusion}
\label{discussion}
\BSg{In this paper we have shown that whether MBSs use the optimal deployment strategy or store ``the most popular content", has very limited impact on the total hit probability if the SBSs are using the optimal deployment strategy when the deployment densities of the SBSs are much larger than the MBSs'. Namely, when MBSs do not use the optimal placement strategy, it is possible to compensate this performance penalty, \ie it is important to optimize the content placement strategy of the SBSs and the total hit probability is increased significantly when the SBSs use the optimal deployment strategy. 

\BSg{For the PPP model, we have defined a new parameter for the probability of storing a file times the deployment density for an individual cache type, and show that the solution is unique for the sum of these new parameters over all cache types. We have shown that the relation between the newly presented parameter and the optimal placement probabilities follow a capacitated transportation problem and we show that the optimal placement probabilities can be obtained greedily. Consecutively, one has the flexibility of choosing the optimal placement strategies of the different types of caches as long as some certain capacity constraint is satisfied.} 

It is shown that heuristic policies for SBSs like storing the popular content that is not yet available in the MBS results in significant performance penalties. We have also proposed \JG{a heuristic} that takes deployment densities of different types of caches into account. We have shown that even though this heuristic gives a better hit probability performance compared to other heuristics, using optimal placement strategy still gives a better hit probability. 

To conclude, using the optimal deployment strategy for the SBSs (typically SBSs have the higher deployment density) is crucial and it ensures the overall network to have the greatest possible total hit probability independent of the deployment policy of MBSs. We have shown that solving the individual problem to find optimal placement strategy for different types of base stations iteratively, namely repeatedly updating the placement strategies of the different types, does not improve the hit probability. Finally, we have shown numerically that LOA gives the same hit probability as the optimal placement strategy of the joint optimization problem of the PPP model by running a single cycle over different types.}

\end{document}